\documentclass[12pt]{article}
\usepackage[english]{babel}
\usepackage{amsmath,amsthm,amsfonts,amssymb,dsfont,mathtools,url}
\usepackage{thm-restate}
\usepackage{enumitem}
\usepackage{xcolor}
\usepackage[backref,colorlinks,citecolor=cyan,bookmarks=true]{hyperref}

\usepackage[margin=1in]{geometry}

\def\eps{\varepsilon}

\setlist[enumerate]{topsep=5pt, itemsep=-0.5ex}
\setlist[itemize]{topsep=0pt, itemsep=-0.5ex}

\newcommand\numberthis{\addtocounter{equation}{1}\tag{\theequation}}

\newcommand\R{\mathbb R}
\newcommand\C{\mathbb C}
\newcommand\F{\mathbb F}

\def\B{\{0,1\}}
\def\pmone{\{-1, 1\}}

\DeclarePairedDelimiter\ceil{\lceil}{\rceil}
\DeclarePairedDelimiter{\abs}{\lvert}{\rvert}
\DeclarePairedDelimiter{\norm}{\lVert}{\rVert}

\DeclareMathOperator*{\E}{\mathbb{E}}                 

\def\xor{\oplus}

\def\polylog{\mathrm{polylog}}

\def\PAD{\textnormal{PAD}}
\def\Many{\textnormal{Many}}

\newtheorem*{theorem*}{Theorem}
\newtheorem{theorem}{Theorem}
\newtheorem{claim}[theorem]{Claim}
\newtheorem{definition}[theorem]{Definition}
\newtheorem{fact}[theorem]{Fact}
\newtheorem{lemma}[theorem]{Lemma}

\newtheorem{corollary}[theorem]{Corollary}

\def\hf{\hat f}
\def\hg{\hat g}

\def\1{\mathds{1}}

\begin{document}
\begin{titlepage}
  \title{Fourier bounds and pseudorandom generators for product tests}
  \author{Chin Ho Lee\footnote{Northeastern University.  Supported by NSF CCF award 1813930.  Work done in part while visiting Amnon Ta-Shma at Tel Aviv University, with support from the Blavatnik family fund and ISF grant no. 952/18.}}
  \maketitle

\begin{abstract}
  We study the Fourier spectrum of functions $f\colon \{0,1\}^{mk} \to \{-1,0,1\}$ which can be written as a product of $k$ Boolean functions $f_i$ on disjoint $m$-bit inputs.
We prove that for every positive integer $d$,
\[
  \sum_{S \subseteq [mk]: |S|=d} |\hat{f_S}| = O(m)^d .
\]
Our upper bound is tight up to a constant factor in the $O(\cdot)$.
Our proof builds on a new ``level-$d$ inequality'' that bounds above $\sum_{|S|=d} \hat{f_S}^2$ for any $[0,1]$-valued function $f$ in terms of its expectation, which may be of independent interest.

As a result, we construct pseudorandom generators for such functions with seed length $\tilde O(m + \log(k/\varepsilon))$, which is optimal up to polynomial factors in $\log m$, $\log\log k$ and $\log\log(1/\varepsilon)$.
Our generator in particular works for the well-studied class of combinatorial rectangles, where in addition we allow the bits to be read in any order.  
Even for this special case, previous generators have an extra $\tilde O(\log(1/\varepsilon))$ factor in their seed lengths. 

Using Schur-convexity, we also extend our results to functions $f_i$ whose range is $[-1,1]$.


\end{abstract}

%
%
%

\thispagestyle{empty}
\end{titlepage}

\section{Introduction}

In this paper we study tests on $n$ bits which can be written as a product of $k$ bounded real-valued functions defined on disjoint inputs of $m$ bits.
We first define them formally.

\begin{definition}[Product tests] \label{def:prod}
  A function $f\colon \B^n \to [-1,1]$ is a \emph{product test} with $k$ functions of input length $m$ if there exist $k$ disjoint subsets $I_1, I_2, \ldots, I_k \subseteq \{1, 2, \ldots, n\}$ of size $\le m$ such that $f(x) = \prod_{i \le k} f_i(x_{I_i})$ for some functions $f_i$ with range in $[-1,1]$.
Here $x_{I_i}$ are the $|I_i|$ bits of $x$ indexed by $I_i$.
\end{definition}

More generally, the range of each function $f_i$ can be $\C_{\le 1} := \{ z \in \C: \abs{z}=1 \}$, the complex unit disk~\cite{GopalanKM15,HLV-bipnfp}, or the set of square matrices over a field~\cite{ReingoldSV13}.
However, in this paper we only focus on the range $[-1,1]$.
As we will soon explain, our results do not hold for the broader range of $\C_{\le 1}$.

The class of product tests was first introduced by Gopalan, Kane and Meka under the name of {\em Fourier shapes}~\cite{GopalanKM15}.
However, in their definition, the subsets $I_i$ are fixed.
Motivated by the recent constructions of pseudorandom generators against {\em unordered} tests, which are tests that read input bits in arbitrary order~\cite{BogdanovPW11,ImpagliazzoMZ12,ReingoldSV13,SteinkeVW14}, Haramaty, Lee and Viola~\cite{HLV-bipnfp} considered the generalization in which the subsets $I_i$ can be arbitrary as long as they are of bounded size and pairwise disjoint. 

Product tests generalize several restricted classes of tests.
For example, when the range of the functions $f_i$ is $\B$, product tests correspond to the AND of disjoint Boolean functions, also known as the well-studied class of {\em combinatorial rectangles}~\cite{AKS87,Nis92,NiZ96,INW94,EvenGLNV98,ArmoniSWZ96,Lu02,Viola-rbd,GopalanMRTV12,GopalanY14}.
When the range of the $f_i$ is $\pmone$, they correspond to the XOR of disjoint Boolean functions, also known as the class of {\em combinatorial checkerboards}~\cite{Watson13}.  
More importantly, product tests also capture {\em read-once space computation}.
Specifically, Reingold, Steinke and Vadhan~\cite{ReingoldSV13} showed that the class of read-once width-$w$ branching programs can be encoded as product tests with outputs $\B^{w \times w}$, the set of $w \times w$ Boolean matrices.

In the past year, the study of product tests~\cite{HLV-bipnfp,LV-rop} has found applications in constructing state-of-the-art pseudorandom generators (PRGs) for space-bounded algorithms.
Using ideas in~\cite{GopalanMRTV12,GopalanY14,LV-rop,ChattopadhyayHRT18}, Meka, Reingold and Tal~\cite{MekaRT18} constructed a pseudorandom generator for width-3 read-once branching programs (ROBPs) on $n$ bits with seed length $\tilde O(\log n \log(1/\eps))$, giving the first improvement of Nisan's generator~\cite{Nis92} in the 90s.
Building on~\cite{ReingoldSV13,HLV-bipnfp,ChattopadhyayHRT18}, Forbes and Kelley significantly simplified the analysis of~\cite{MekaRT18} and constructed a generator that fools {\em unordered} polynomial-width read-once branching programs.
Thus, it is motivating to further study product tests, in the hope of gaining more insights into constructing better generators for space-bounded algorithms, and resolving the long-standing open problem of RL vs.\ L.

In this paper we are interested in understanding the Fourier spectrum of product tests.
We first define the {\em Fourier weight} of a function.
For a function $f\colon \B^n \to \R$, consider its Fourier expansion $f = \sum_{S \subseteq [n]} \hf_S \chi_S$.

\newpage

\begin{definition}[$d$th level Fourier weight in $L_q$-norm] \label{def:Fourierweight}
  Let $f\colon \B^n \to \C_{\le 1}$ be any function.
  The {\em $d$th level Fourier weight of $f$ in $L_q$-norm} is
  \[
    W_{q,d}[f] := \sum_{\abs{S}=d} \abs{\hf_S}^q .
  \]
  We denote by $W_{q,\le d}[f]$ the sum $\sum_{\ell=0}^d W_{q,\ell}[f]$.
\end{definition}

Several papers have studied the Fourier spectrum of different classes of tests.
This includes constant-depth circuits~\cite{Mansour95,Tal17},
read-once branching programs~\cite{ReingoldSV13,SteinkeVW14,ChattopadhyayHRT18}, and low-sensitivity functions~\cite{GopalanRW16}.
More specifically, these papers showed that they have {\em bounded $L_1$ Fourier tail},  that is, there exists a positive number $b$ such that for every test $f$ in the class and every positive integer $d$, we have
\[
  W_{1,d}[f] \le b^d .
\]
One technical contribution of this paper is giving tight upper and lower bounds on the $L_1$ Fourier tail of product tests.

\begin{restatable}{theorem}{loneprodbool} \label{thm:L1prod-Bool}
  Let $f\colon \B^n \to [-1,1]$ be a product test of $k$ functions $f_1, \ldots, f_k$ with input length $m$.
  Suppose there is a constant $c > 0$ such that $\abs{\E[f_i]} \le 1 - 2^{-cm}$ for every $f_i$.
  For every positive integer $d$, we have
  \[
    W_{1,d}[f] \le \bigl( 72 (\sqrt{c} \cdot m) \bigr)^d .
  \]
\end{restatable}
Theorem~\ref{thm:L1prod-Bool} applies to Boolean functions $f_i$ with outputs $\B$ or $\pmone$.
Moreover, the parity function on $mk$ bits can be written as a product test with outputs $\pmone$, which has $\hf_{[mk]} = 1$.
So product tests do not have non-trivial $L_2$ Fourier tail.
(See~\cite{Tal17} for a definition.)

We also obtain a different upper bound when the $f_i$ are arbitrary $[-1,1]$-valued functions.

\begin{restatable}{theorem}{loneprod} \label{thm:L1prod}
  Let $f\colon \B^n \to [-1,1]$ be a product test of $k$ functions $f_1, \ldots, f_k$ with input length $m$.
  Let $d$ be a positive integer.
  We have
  \[
    W_{1,d}[f] \le \bigl( 85\sqrt{m \ln(4ek)} \bigr)^d .
  \]
\end{restatable}

We note that Theorems~\ref{thm:L1prod-Bool} and~\ref{thm:L1prod} are incomparable, as one can take $m = 1$ and $k = n$, or $m = n$ and $k = 1$. 

\begin{restatable}{claim}{lowerbound} \label{claim:lowerbound}
  For all positive integers $m$ and $d$, there exists a product test $f\colon \B^{mk} \to \B$ with $k = d \cdot 2^m$ functions of input length $m$ such that
  \[
    W_{1,d}[f] \ge (m/e^{3/2})^d .
  \]
\end{restatable}
This matches the upper bound $W_{1,d}[f] = O(m)^d$ in Theorem~\ref{thm:L1prod-Bool} up to the constant in the $O(\cdot)$.
Moreover, applying Theorem~\ref{thm:L1prod} to the product test $f$ in Claim~\ref{claim:lowerbound} gives $W_{1,d}[f] = O(\sqrt{m \log(2k)})^d = O(m + \sqrt{m \log d})^d$.
Therefore, for all integers $m$ and $d \le 2^{O(m)}$, there exists an integer $k$ and a product test $f$ such that the upper bound $W_{1,d}[f] = O(\sqrt{m \log(2k)})^d$ is tight up to the constant in the $O(\cdot)$.

We now discuss some applications of Theorems~\ref{thm:L1prod-Bool} and~\ref{thm:L1prod} in pseudorandomness.

\paragraph{Pseudorandom generators.}
In recent years, researchers have developed new frameworks to construct pseudorandom generators against different classes of tests.
Gopalan, Meka, Reingold, Trevisan and Vadhan~\cite{GopalanMRTV12} refined a framework introduced by Ajtai and Wigderson~\cite{AjtaiW89} to construct better generators for the classes of combinatorial rectangles and read-once DNFs.
Since then, this framework has been used extensively to construct new PRGs against different classes of tests~\cite{TrevisanX13,GopalanKM15,GopalanY14,ReingoldSV13,SteinkeVW14,ChenSV15,HLV-bipnfp,HatamiT18,ServedioT18,LV-rop,ChattopadhyayHRT18,ForbesK18,MekaRT18,DoronHH18}.
Recently, a beautiful work by Chattopadhyay, Hatami, Hosseini and Lovett~\cite{ChattopadhyayHHL18} developed a new framework of constructing PRGs against any classes of functions that are closed under restriction and have bounded $L_1$ Fourier tail.
Thus, applying their result to Theorems~\ref{thm:L1prod-Bool} and~\ref{thm:L1prod}, we can immediately obtain a non-trivial PRG for product tests.
However, using the recent result of Forbes and Kelley~\cite{ForbesK18} and exploiting the structure of product tests, we use the Ajtai--Wigderson framework to construct PRGs with much better seed length than using~\cite{ChattopadhyayHHL18} as a blackbox.

\begin{restatable}{theorem}{prgxor} \label{thm:prg-xor}
  There exists an explicit generator $G\colon \B^\ell \to \B^n$ that fools the XOR of any $k$ Boolean functions on disjoint inputs of length $\le m$ with error $\eps$ and seed length $O(m + \log(n/\eps))(\log m + \log\log(n/\eps))^2 = \tilde O(m + \log(n/\eps))$.
\end{restatable}

Here $\tilde O(1)$ hides polynomial factors in $\log m$, $\log\log k$, $\log\log n$ and $\log\log(1/\eps)$.
When $mk = n$ or $\eps = n^{-\Omega(1)}$, the generator in Theorem~\ref{thm:prg-xor} has seed length $\tilde O(m + \log(k/\eps))$, which is optimal up to $\tilde O(1)$ factors.

We now compare Theorem~\ref{thm:prg-xor} with previous works.
Using a completely different analysis, Lee and Viola~\cite{LV-rop} obtained a generator with seed length $\tilde O((m + \log k))
 \log(1/\eps)$.  
 When $m = O(\log n)$ and $k = 1/\eps = n^{\Omega(1)}$, this is $\tilde O(\log^2 n)$, whereas the generator in Theorem~\ref{thm:prg-xor} has seed length $\tilde O(\log n)$.
When each function $f_i$ is computable by a read-once width-$w$ branching program on $m$ bits, Meka, Reingold and Tal~\cite{MekaRT18} obtained a PRG with seed length $O(\log (n/\eps)) (\log m + \log\log (n/\eps))^{2w+2}$.
When $m = O(\log(n/\eps))$, Theorem~\ref{thm:prg-xor} improves on their generator on the lower order terms.
As a result, we obtain a PRG for {\em read-once $\F_2$-polynomials}, which are a sum of monomials on disjoint variables over $\F_2$, with seed length $O(\log n/\eps) (\log\log (n/\eps))^2$.
This also improves on the seed length of their PRG for read-once polynomials in the lower order terms by a factor of $(\log\log(n/\eps))^4$.

Our generator in Theorem~\ref{thm:prg-xor} also works for the AND of the functions $f_i$, corresponding to the class of {\em unordered} combinatorial rectangles.
In fact, we have the following more general corollary.

\begin{corollary} \label{cor:prg-L1}
  There exists an explicit pseudorandom generator $G\colon \B^\ell \to \B^n$ with seed length $\tilde O(m + \log(n/\eps))$ such that the following holds.
  Let $f_1, \ldots, f_k\colon \B^{I_i} \to \B$ be $k$ Boolean functions where the subsets $I_i \subseteq [n]$ are pairwise disjoint and have size at most $m$.
  Let $g\colon \B^k \to \C_{\le 1}$ be any function and write $g$ in its Fourier expansion $g = \sum_{S \subseteq [k]} \hg_S \chi_S$.
  Then $G$ fools $g(f_1, \ldots, f_k)$ with error $L_1[g] \cdot \eps$,
  where $L_1[g] := \sum_{S \neq \emptyset} \abs{\hg_S}$.
\end{corollary}
\begin{proof}
  Let $G$ be the generator in Theorem~\ref{thm:prg-xor}.
  Note that $\chi_S(f_1(x_{I_1}), \ldots, f_k(x_{I_k}))$ is a product test with outputs $\pmone$.
    So by Theorem~\ref{thm:prg-xor} we have
    \begin{align*}
      &\quad \bigl| \E[g(f_1(U_{I_1}), \ldots, f_k(U_{I_k})) - \E[g(f_1(G_{I_1}), \ldots, f_k(G_{I_k})] \bigr| \\
      &\le \sum_S \abs{\hg_S} \bigl| \E[\chi_S(f_1(U_{I_1}), \ldots, f_k(U_{I_k}))] - \E[\chi_S(f_1(G_{I_1}), \ldots, f_k(G_{I_k})] \bigr| \\
      &\le L_1[g] \cdot \eps . \qedhere
    \end{align*}
\end{proof}
Note that the AND function has $L_1[\mathrm{AND}] \le 1$, and so the generator in Corollary~\ref{cor:prg-L1} fools unordered combinatorial rectangles.
Previous generators for unordered combinatorial rectangles use almost-bounded independence or small-bias distributions, and have seed length $O(\log(n/\eps)) (1/\eps)$~\cite{ChariRS00,DeETT10}.

When the functions $f_i$ in the product tests have outputs $[-1,1]$, we also obtain the following generator.
\begin{restatable}{theorem}{prgprod} \label{thm:prg-prod}
There exists an explicit generator $G\colon \B^\ell \to \B^n$ that fools any product test with $k$ functions of input length $m$ with error $\eps$ and seed length $O(m + \log(k/\eps))\log(k/\eps)\allowbreak(\log m + \log\log n) = \tilde O(m + \log(k/\eps)) \log(k/\eps)$.
\end{restatable}
When $m = o(\log n)$ and $k = 1/\eps = 2^{o(\sqrt{\log n})}$, Theorem~\ref{thm:prg-prod} gives a better seed length than Theorem~\ref{thm:prg-xor}.
Thus the generator in Theorem~\ref{thm:prg-prod} remains interesting for $f_i \in \pmone$ when a product test $f$ depends on very few variables and the error $\eps$ is not so small.

Previous best generator~\cite{LV-rop} has an extra $\tilde O(\log (1/\eps))$ in the seed length.
However, the generator in~\cite{LV-rop} works even when the $f_i$ have range $\C_{\le 1}$, which implies generators for several variants of product tests such as generalized halfspaces and combinatorial shapes.
(See~\cite{GopalanKM15} for the reductions.)

Finally, when the subsets $I_i$ of a product test are fixed and known in advanced, Gopalan, Kane and Meka~\cite{GopalanKM15} constructed a PRG of the same seed length as Theorem~\ref{thm:prg-xor}, but again their PRG works more generally for the range of $\C_{\le 1}$ instead of $\pmone$.

\paragraph{$\F_2$-polynomials.}
Chattopadhyay, Hatami, Lovett and Tal~\cite{ChattopadhyayHLT19} recently constructed a pseudorandom generator for any class of functions that are closed under restriction, provided there is an upper bound on the second level Fourier weight of the functions in $L_1$-norm.
They conjectured that every $n$-variate $\F_2$-polynomial $f$ of degree $d$ satisfies the bound $W_{1,2}[f] = O(d^2)$.
In particular, a bound of $n^{1/2-o(1)}$ would already imply a generator for polynomials of degree $d = \Omega(\log n)$, a major breakthrough in complexity theory.
Theorem~\ref{thm:L1prod} shows that their conjecture is true for the special case of {\em read-once} polynomials.
In fact, it shows that $W_{1,t}[f] = O(d^t)$ for every positive integer $t$.
Previous bound for read-once polynomials gives $W_{1,t}[f] = O(\log^4 n)^t$~\cite{ChattopadhyayHRT18}.

\paragraph{The coin problem.}

Let $X_{n,\eps} = (X_1, \ldots, X_n)$ be the distribution over $n$ bits, where the variables $X_i$ are independent and each $X_i$ equals $1$ with probability $(1 - \eps)/2$ and $0$ otherwise.
The {\em $\eps$-coin problem} asks whether a given function $f$ can distinguish between the distributions $X_{n,\eps}$ and $X_{n,0}$ with advantage $1/3$.

This central problem has wide range of applications in computational complexity and has been studied extensively for different restricted classes of tests, including bounded-depth circuits~\cite{Ajt83,Valiant84-Majority,AjB84,Amano09,ViolaBPvsE,ShV-dec,Aaronson10,Viola-rbd,CohenGR14}, space-bounded algorithms~\cite{BrodyV10,Steinberger13,CohenGR14}, bounded-depth circuits with parity gates~\cite{ShV-dec,KoppartyS18,RossmanS17,LimayeSSTV18}, $\F_2$-polynomials~\cite{LimayeSSTV18,ChattopadhyayHLT19} and product tests~\cite{LeeV-coin}.

It is known that if a function $f$ has bounded $L_1$ Fourier tail, then it implies a lower bound on the smallest $\eps^*$ of $\eps$ that $f$ can solve the $\eps$-coin problem.

\begin{fact} \label{fact:L1impliescoin}
  Let $f\colon \B^n \to \C_{\le 1}$ be any function.
  If for every integer $d \in \{0, \ldots, n\}$ we have $W_{1,d}[f] \le b^d$, then $f$ solves the $\eps$-coin problem with advantage at most $2 b \eps$.
\end{fact}
\begin{proof}
  We may assume $b \eps \le 1/2$, otherwise the result is trivial.
  Observe that we have $\E[\chi_S(X_{n,\eps})] = \eps^{\abs{S}}$ for every subset $S \subseteq [n]$.
  Thus,
  \begin{multline*}
    \bigl| \E[f(X_{n,\eps})] - \E[f(X_{n,0})] \bigr|
    = \Bigl| \sum_{S \ne \emptyset} \hf_S \E[X_{n,\eps}] \Bigr| \\
    \le \sum_{d=1}^n \sum_{\abs{S}=d} \abs{\hf_S} \cdot \eps^d 
    = \sum_{d=1}^n (b \eps)^d
    \le b\eps \cdot \sum_{d=1}^n 2^{-(d-1)} 
    \le 2 b\eps . \qedhere
  \end{multline*}
\end{proof}

Lee and Viola~\cite{LeeV-coin} showed that product tests with range $[-1,1]$ can solve the $\eps$-coin problem with $\eps^* = \Theta(1/\sqrt{m \log k})$.
Hence, Fact~\ref{fact:L1impliescoin} implies that Theorem~\ref{thm:L1prod} recovers their lower bound. 
Moreover, their upper bound implies that the dependence on $m$ and $k$ in Theorem~\ref{thm:L1prod} is tight up to constant factors when $d$ is constant.
Claim~\ref{claim:lowerbound} complements this by showing that the dependence on $d$ in Theorem~\ref{thm:L1prod} is also tight for some choice of $k$.

The work~\cite{LeeV-coin} also shows that when the range of the functions $f_i$ is $\C_{\le 1}$, the right answer for $\eps^*$ is $\Theta(1/\sqrt{mk})$.
Therefore, one cannot obtain for a better tail bound than the trivial bound of $(\sqrt{mk})^d$ when the range is $\C_{\le 1}$.

\subsection{Techniques}

We now explain how to obtain Theorems~\ref{thm:L1prod-Bool} and~\ref{thm:L1prod} and our pseudorandom generators for product tests (Theorems~\ref{thm:prg-xor} and~\ref{thm:prg-prod}).

\subsubsection{Fourier spectrum of product tests}
The high-level idea of proving Theorems~\ref{thm:L1prod-Bool} and~\ref{thm:L1prod} is inspired from~\cite{LeeV-coin}.
For intuition, let us first assume that the functions $f_i$ have outputs $\B$ and are all equal to $f_1$ (but defined on disjoint inputs).
It will also be useful to think of the number of functions $k$ being much larger than input length $m$ of each function.
We first explain how to bound above $W_{1,1}[f]$.
(Recall in Definition~\ref{def:Fourierweight} we defined $W_{q,d}[f]$ of a function $f$ to be $\sum_{\abs{S}=d} \abs{\hf_S}^q$.)

\paragraph{Bounding $W_{1,1}[f]$.}
Since the functions $f_i$ of a product test $f$ are defined on disjoint inputs, each Fourier coefficient of $f$ is a product of the coefficients of the $f_i$, and so each weight-$1$ coefficent of $f$ is a product of $k-1$ weight-$0$ and $1$ weight-$1$ coefficients of the $f_i$.
From this, we can see that $W_{1,1}[f]$ is equal to 
\[
  \binom{k}{1} \cdot W_{1,1}[f_1] \cdot W_{1,0}[f_1]^{k-1} = k \cdot W_{1,1}[f_1] \cdot \E[f_1]^{k-1} \numberthis \label{eqn:intro1}.
\]
Because of the term $\E[f_1]^{k-1}$, to maximize $W_{1,1}[f]$ it is natural to consider taking $f_1$ to be a function with expectation $\E[f_1]$ as close to $1$ as possible, i.e.\ the OR function.
In such case, one would hope for a better bound on $W_{1,1}[f_1]$.
Indeed, Chang's inequality~\cite{Chang02} (see also~\cite{ImpagliazzoMR14} for a simple proof) says that for a $[0,1]$-valued function $g$ with expectation $\alpha \le 1/2$, we have 
\[
  W_{2,1}[g] \le 2 \alpha^2 \ln(1/\alpha) .
\]
(The condition $\alpha \le 1/2$ is without loss of generality as one can instead consider $1 - g$.)
It follows by a simple application of the Cauchy--Schwarz inequality that $W_{1,1}[g] \le O(\sqrt{n})\nolinebreak\cdot\nolinebreak\alpha \sqrt{\ln(1/\alpha)}$ (see Fact~\ref{fact:L1fromL2} below for a proof).
Moreover, when the functions $f_i$ are Boolean, we have $2^{-m} \le \E[f_i] \le 1 - 2^{-m}$, and so $\sqrt{\ln(1/\alpha)} \le \sqrt{m}$.
Plugging these bounds into Equation~\eqref{eqn:intro1}, we obtain a bound of $O(m) \cdot k (1 - \E[f_1]) \E[f_1]^{k-1}$.
So indeed $\E[f_1]$ should be roughly $1 - 1/k$ in order to maximize $W_{1,1}[f]$, giving an upper bound of $O(m)$.
For the case where the $f_i$ can be different, a simple convexity argument shows that $W_{1,1}[f]$ is maximized when the functions $f_i$ have the same expectation.

\paragraph{Bounding $W_{1,d}[f]$ for $d > 1$.}
To extend this argument to $d > 1$, one has to generalize Chang's inequality to bound above $W_{2,d}[g]$ for $d > 1$.
The case $d=2$ was already proved by Talagrand~\cite{Talagrand96}.
Following Talagrand's argument in~\cite{Talagrand96} and inspired by the work of Keller and Kindler~\cite{KellerK13}, which proved a similar bound in terms of a different measure than $\E[g]$, we prove the following bound on $W_{2,d}[g]$ in terms of its expectation.

\begin{restatable}{lemma}{levelk}\label{lemma:level-k-inequalities}
  Let $g\colon \B^n \to [0,1]$ be any function.
  For every positive integer $d$, we have
  \[
    W_{2,d}[g] \le 4 \E[g]^2 \bigl( 2e \ln(e/\E[g]^{1/d}) \bigr)^d .
  \]
\end{restatable}

We note that the exponent $1/d$ of $\E[g]$ either did not appear in previous upper bounds (mentioned without proof in~\cite{ImpagliazzoMR14}), or only holds for restricted values of $d$~\cite{ODonnell14}.
This exponent is not important for proving Theorem~\ref{thm:L1prod-Bool}
, but will be crucial in the proof of Theorem~\ref{thm:L1prod}, which we will explain later on.

For $d > 1$, the expression for $W_{1,d}[f]$ becomes much more complicated than $W_{1,1}[f]$, as it involves $W_{1,z}[f_1]$ for different values of $z \in [m]$.
So one has to formulate the expression of $W_{1,d}[f]$ carefully.
(See Lemma~\ref{lemma:L1prod}.)
Once we have obtained the right expression for $W_{1,d}[f]$, the proof of Theorem~\ref{thm:L1prod-Bool} follows the outline above by replacing Chang's inequality with Lemma~\ref{lemma:level-k-inequalities}.
One can then handle functions $f_i$ with outputs $\pmone$ by considering the translation $f_i \mapsto (1-f_i)/2$, which only changes each $W_{1,d}[f_i]$ (for $d > 0$) by a factor of $2$.
We remark that Theorem~\ref{thm:L1prod-Bool} is sufficient for constructing the generator in Theorem~\ref{thm:prg-xor}.

\paragraph{Handling $[-1,1]$-valued $f_i$.}
Extending this argument to proving Theorem~\ref{thm:L1prod} poses several challenges.
Following the outline above, after plugging in Lemma~\ref{lemma:level-k-inequalities}, 
we would like to show that $\E[f_1]$ should be roughly $1 - 1/k$ to maximize $W_{1,d}[f]$.
However, it is no longer clear why this is the case even assuming the maximum is attained by functions $f_i$ with the same expectation, as we now do not have the bound $\sqrt{\ln(1/\alpha)} \le \sqrt{m}$, and so it cannot be used to simplify the expression of $W_{1,d}[f]$ as before.
In fact, the above assumption is simply false if we plug in the upper bound in Lemma~\ref{lemma:level-k-inequalities} with the exponent $1/d$ omitted to the $W_{1,z_i}[f_i]$.

Using Lemma~\ref{lemma:level-k-inequalities} and the symmetry of the expression for $W_{1,d}[f]$, we reduce the problem of bounding above $W_{1,d}[f]$ with different $f_i$ to bounding the same quantity but with the additional assumption that the $f_i$ have the same expectation $\E[f_1]$.
This uses Schur-convexity (see Section~\ref{sec:L1prod} for its definition).
Then by another convexity argument we show that the maximum is attained when $\E[f_1]$ is roughly equal to $1 - d/k$.
Both of these arguments critically rely on the aforementioned exponent of $1/d$ in Lemma~\ref{lemma:level-k-inequalities}.

%
\subsubsection{Pseudorandom generators}
We now discuss how to use Theorems~\ref{thm:L1prod-Bool} and~\ref{thm:L1prod} to construct our pseudorandom generators for product tests.
Our construction follows the Ajtai--Wigderson framework~\cite{AjtaiW89} that was recently revived and refined by Gopalan, Meka, Reingold, Trevisan and Vadhan~\cite{GopalanMRTV12}.

The high-level idea of this framework involves two steps.
For the first step, we show that {\em derandomized bounded independence plus noise} fools $f$.
More precisely, we will show that if we start with a small-bias or almost-bounded independent distribution $D$ (``bounded independence''), and select roughly half of $D$'s positions $T$ pseudorandomly and set them to uniform $U$ (``plus noise''),
then this distribution, denoted by $D + T \wedge U$, fools product tests.

Forbes and Kelley~\cite{ForbesK18} recently improved the analysis in~\cite{HLV-bipnfp} and implicitly showed that $\delta$-almost $d$-wise independent plus noise fools product tests, where $d = O(m + \log(k/\eps))$ and $\delta = n^{-\Omega(d)}$.
Using Theorem~\ref{thm:L1prod}, we improved the dependence on $\delta$ to $(m \ln k)^{-\Omega(d)}$ and obtain the following theorem.

\begin{restatable}{theorem}{bipnfp} \label{thm:bipnfp}
  Let $f\colon \B^n \to [-1,1]$ be a product test with $k$ functions of input length $m$.
  Let $d$ be a positive integer.
  Let $D$ and $T$ be two independent $\delta$-almost $d$-wise independent distributions over $\B^n$, and $U$ be the uniform distribution over $\B^n$.
  Then
  \[
    \bigl| \E[f(D + T \wedge U)] - \E[f(U)] \bigr| \le k \cdot \bigl( \sqrt{\delta} \cdot (170 \cdot \sqrt{m \ln(ek)})^d + 2^{-(d-m)/2} \bigr) ,
  \]
  where ``$+$'' and ``$\wedge$'' are bit-wise XOR and AND respectively.
\end{restatable}


The second step of the Ajtai--Wigderson framework builds a pseudorandom generator by applying the first step (Theorem~\ref{thm:bipnfp}) recursively.
Let $f\colon \B^n \to \B$ be a product test with $k$ functions of input length $m$.
As product tests are closed under restrictions (and shifts), 
after applying Theorem~\ref{thm:bipnfp} to $f$ and fixing $D$ and $T$ in the theorem, the function $f_{D,T}\colon \B^T \to \B$ defined by $f_{D,T}(y) := f(D + T \wedge y)$ is also a product test.
Thus one can apply Theorem~\ref{thm:bipnfp} to $f_{D,T}$ again and repeat the argument recursively.
We will use different progress measures to bound above the number of recursion steps in our constructions.
We first describe the recursion in Theorem~\ref{thm:prg-prod} as it is simpler.

\paragraph{Fooling $[-1,1]$-valued product tests.}
Here our progress measure is the maximum input length $m$ of the functions $f_i$.
We show that after $O(\log(k/\eps))$ steps of the recursion, the functions $f_i$ of the restricted product test have their input length halved with high probability.
Therefore, repeating above for $O(\log m)$ steps, the product test is restricted to a constant function.
This simple recursion gives our second PRG (Theorem~\ref{thm:prg-prod}).

\paragraph{Fooling Boolean-valued product tests.}
Our construction of the first generator (Theorem~\ref{thm:prg-xor}) is more complicated and uses two progress measures.
The first one is again the maximum input length $m$ of the functions $f_i$, and the second is the number $k$ of the functions $f_i$.
We reduce the number of recursion steps from $O(\log (k/\eps)) \log m$ to $O(\log m)$. 
This requires a more delicate construction and analysis that are similar to the recent work of Meka, Reingold and Tal~\cite{MekaRT18}, which constructed a pseudorandom generator against XOR of disjoint constant-width read-once branching programs.
There are two main ideas in their construction.
First, they ensure $k \le 16^m$ in each step of the recursion, by constructing another PRG to fool the test $f$ for the case $k \ge 16^m$.
We will also use this PRG in our construction.
Next, throughout the recursion they allow one ``bad'' function $f_i$ of the product test $f$ to have a longer input length than $m$, but not longer than $O(\log(n/\eps))$.
Using these two ideas, they show that whenever $m \ge \log\log n$ during the recursion, then after $O(1)$ steps of the recursion all but the ``bad'' $f_i$ have their input length restricted by a half, while the ``bad'' $f_i$ always has length $O(\log (n/\eps))$.
This allows us to repeat $O(\log m)$ steps until we are left with a product test of $k' \le \polylog(n)$ functions, where all but one of the $f_i$ have input length at most $m' = O(\log \log n)$.

Now we switch our progress measure to the number of functions.
This part is different from~\cite{MekaRT18}, in which their construction relies on the fact that the $f_i$ are computable by read-once branching programs.
Here because our functions $f_i$ are arbitrary, by grouping $c$ functions as one, we can instead think of the parameters $k'$ and $m'$ in the product test as $k''=k'/c$ and $m''=cm'$, respectively.
Choosing $c$ to be $O(\log n/\log\log n)$, we have $m'' = O(\log n)$ and so we can repeat the previous argument again.
Because each time $k'$ is reduced by a factor of $c$, after repeating this for $O(1)$ steps, we are left with a product test defined on $O(\log n)$ bits, which can be fooled using a small-bias distribution.
This gives our first generator (Theorem~\ref{thm:prg-xor}).

\paragraph{Organization}
In Section~\ref{sec:L1prod} we prove Theorems~\ref{thm:L1prod-Bool} and \ref{thm:L1prod}.
In Section~\ref{sec:prg} we construct our pseudorandom generators for product tests, proving Theorems~\ref{thm:prg-xor} and~\ref{thm:prg-prod}.
In Section~\ref{sec:levelk} we prove Lemma~\ref{lemma:level-k-inequalities}, which is used in the proof of Theorem~\ref{thm:L1prod}.

\section{Fourier spectrum of product tests} \label{sec:L1prod}

In this section we prove Theorems~\ref{thm:L1prod-Bool} and~\ref{thm:L1prod}.
We first restate the theorems.

\loneprodbool*

\loneprod*

Both theorems rely on the following lemma which gives an upper bound on $W_{2,d}[g]$ in terms of the expectation of a $[0,1]$-valued function $g$.
The case $d=1$ is known as Chang's inequality~\cite{Chang02}.
(See also~\cite{ImpagliazzoMR14} for a simple proof.)
This was then generalized by Talagrand to $d=2$~\cite{Talagrand96}.
Using a similar argument to~\cite{Talagrand96}, we extend this to $d > 2$.

\levelk*

We defer its proof to Section~\ref{sec:levelk}.
We remark that a similar upper bound was proved by Keller and Kindler~\cite{KellerK13}.
However, the upper bound in~\cite{KellerK13} was proved in terms of $\sum_{i=1}^n I_i[g]^2$, where $I_i[g]$ is the influence of the $i$th coordinate on $g$, instead of $\E[g]$.
A similar upper bound in terms of $\E[g]$ can be found in~\cite{ODonnell14} under the extra condition $d \le 2 \ln (1/\E[g])$.

We will also use the following well-known fact that bounds above $W_{1,d}[f]$ in terms of $W_{2,d}[f]$.

\begin{fact} \label{fact:L1fromL2}
  Let $f\colon \B^n \to \R$ be any function.
	We have $W_{1,d}[f] \le n^{d/2} \sqrt{W_{2,d}[f]}$.
\end{fact}
\begin{proof}
  By the Cauchy--Schwarz inequality, 
  \[
    W_{1,d}[f]
    = \sum_{\abs{S}=d} \abs{\hf_S}
    \le \sqrt{\binom{n}{d} \sum_{\abs{S}=d} \hf_S^2}
    \le n^{d/2} \sqrt{W_{2,d}[f]} . \qedhere
  \]
\end{proof}

\begin{lemma} \label{lemma:L1prod}
  Let $f\colon \B^n \to [-1,1]$ be a product test of $k$ functions $f_1, \ldots, f_k$ with input length $m$, and $\alpha_i := (1 - \E[f_i])/2$ for every $i \in [k]$.
  Let $d$ be a positive integer.
  We have
  \[
    W_{1,d}[f]
    \le \bigl(\sqrt{32e^3 m} \bigr)^d g(\alpha_1, \ldots, \alpha_k) ,
  \]
  where the function $g\colon (0,1]^k \to \R$ is defined by
  \[
    g(\alpha_1, \ldots, \alpha_k)
    := e^{-2 \sum_{i=1}^k \alpha_i} \sum_{\ell=1}^d \sum_{\substack{S \subseteq [k] \\ \abs{S}=\ell}} \sum_{\substack{z \in [m]^S \\ \sum_i z_i = d}} \prod_{i \in S} \Bigl( \alpha_i \bigr( \ln \bigl( e/\alpha_i^{1/z_i} \bigr) \bigr)^{z_i/2} \Bigr) .
  \]
\end{lemma}

\begin{proof}
  For notational simplicity, we will use $W_d[f]$ to denote $W_{1,d}[f]$.
  Write $f = \prod_{i=1}^k f_i$.
  Without loss of generality we will assume each function $f_i$ is non-constant.
  Since $f_i$ and $-f_i$ have the same weight $W_d[f_i]$, we will further assume $\E[f_i] \in [0,1)$.
  Note that for a subset $S = S_1 \times \cdots \times S_k \subseteq (\{0,1\}^m)^k$, we have $\hf_S = \prod_{i=1}^k \hf_{i_{S_i}}$.
  So,
  \[
    W_d[f]
    = \sum_{\substack{z \in \{0,\ldots,m\}^k \\ \sum_i z_i = d}} \prod_{i=1}^k W_{z_i}[f_i]
    = \sum_{\ell=1}^d \sum_{\substack{S \subseteq [k]\\ \abs{S}=\ell}} \sum_{\substack{z \in [m]^S \\ \sum_i z_i=d}}
    \Bigl( \prod_{i\in S} W_{z_i}[f_i] \cdot \prod_{i\not\in S} W_0[f_i] \Bigr) .
  \]
  Since $x = 1 - (1 - x) \le e^{-(1-x)}$ for every $x \in \R$, for every subset $S \subseteq [k]$ of size at most $d$, we have
  \[
    \prod_{i \not\in S} W_{z_i}[f_i]
    \le e^{-\sum_{i \not\in S} (1 - W_{z_i}[f_i])}
    \le e^{-\sum_{i \not\in S} (1 - W_{z_i}[f_i])} \cdot e^{\sum_{i \in S} W_{z_i}[f_i]}
    \le e^d \cdot e^{-\sum_{i=1}^k (1 - W_{z_i}[f_i])} .
  \]
  Hence,
  \begin{align*}
    W_d[f]
    &= \sum_{\ell=1}^d \sum_{\substack{S \subseteq [k]\\ \abs{S}=\ell}} \sum_{\substack{z \in [m]^S \\ \sum_i z_i=d}} \Bigl( \prod_{i\in S} W_{z_i}[f_i] \cdot \prod_{i\not\in S} W_0[f_i] \Bigr) \\
    &\le e^d \cdot e^{- \sum_{i=1}^k (1 - W_0[f_i])} \sum_{\ell=1}^d \sum_{\substack{S \subseteq [k]\\ \abs{S}=\ell}} \sum_{\substack{z \in [m]^S \\ \sum_i z_i=d}} \prod_{i\in S} W_{z_i}[f_i] . \numberthis \label{eqn:L1proda}
  \end{align*}
  Define $f_i' := (1-f_i)/2 \in [0,1]$.
  Let $\alpha_i := \E[f_i'] = (1 - \E[f_i])/2 \in (0, 1/2]$.
  Applying Lemma~\ref{lemma:level-k-inequalities} and Fact~\ref{fact:L1fromL2} to the functions $f_i'$, we have for every subset $S \subseteq [k]$ of size at most $d$,
  \begin{align*}
    \sum_{\substack{z \in [m]^S \\ \sum_i z_i=d}} \prod_{i\in S} W_{z_i}[f_i']
    &\le  \sum_{\substack{z \in [m]^S \\ \sum_i z_i=d}} \prod_{i\in S} \Bigl( 2 m^{z_i/2} \alpha_i \bigl( 2e \ln \bigl(e/\alpha_i^{1/z_i} \bigr) \bigr)^{z_i/2} \Bigr) \\
    &\le (\sqrt{8em})^d \sum_{\substack{z \in [m]^S \\ \sum_i z_i=d}} \prod_{i\in S} \Bigl( \alpha_i \bigl( \ln \bigl(e/\alpha_i^{1/z_i} \bigr) \bigr)^{z_i/2} \Bigr) .
  \end{align*}
  Note that for every integer $d \ge 1$, we have $W_d[f_i] = 2W_d[f_i']$.
  Plugging the bound above into Equation~\eqref{eqn:L1proda}, we have
  \[
    W_d[f]
    \le (2e)^d \cdot e^{- 2\sum_{i=1}^k \alpha_i} \sum_{\ell=1}^d \sum_{\substack{S \subseteq [k]\\ \abs{S}=\ell}} \sum_{\substack{z \in [m]^S \\ \sum_i z_i=d}} \prod_{i\in S} W_{z_i}[f_i'] 
    \le \bigl(\sqrt{32e^3 m} \bigr)^d g(\alpha_1, \ldots, \alpha_k) ,
  \]
  where the function $g\colon (0,1]^k \to \R$ is defined by
  \[
    g(\alpha_1, \ldots, \alpha_k)
    := e^{-2 \sum_{i=1}^k \alpha_i} \sum_{\ell=1}^d \sum_{\substack{S \subseteq [k] \\ \abs{S}=\ell}} \sum_{\substack{z \in [m]^S \\ \sum_i z_i = d}} \prod_{i \in S} \Bigl( \alpha_i \bigr( \ln \bigl( e/\alpha_i^{1/z_i} \bigr) \bigr)^{z_i/2} \Bigr) . \qedhere
  \]
  \end{proof}

  We now prove Theorems~\ref{thm:L1prod-Bool} and~\ref{thm:L1prod}.
  For every $(\alpha_1, \ldots, \alpha_k) \in (0,1]^k$, let $\alpha := \sum_{i=1}^k \alpha_i/k \in (0,1]$.
  We note that the upper bound in Theorem~\ref{thm:L1prod-Bool} is sufficient to prove Theorem~\ref{thm:prg-xor}.

  \begin{proof}[Proof of Theorem~\ref{thm:L1prod-Bool}]
  We will bound above $g(\alpha_1, \ldots, \alpha_k)$ in Lemma~\ref{lemma:L1prod}.
  Recall that $\alpha_i = (1 - \E[f_i])/2$.
  Since $\abs{\E[f_i]} \le 1 - 2^{-cm}$, we have $\alpha_i \ge 2^{-(cm+1)}$, and so $\ln(1/\alpha_i) \le cm+1$.
  For every subset $S \subseteq [k]$, the set $\{z \in [m]^S: \sum_i z_i = d\}$ has size at most $\binom{d-1}{\abs{S}-1} \le 2^d$.
  Hence, 
  \[
    \sum_{\substack{z \in [m]^S \\ \sum_i z_i=d}} \prod_{i\in S} \bigl( \ln(1/\alpha_i) \bigr)^{z_i/2}
    \le 2^d (cm+1)^{d/2} .
  \]
  By Maclaurin's inequality (cf.~\cite[Chapter~12]{Steele04}), we have
  \[
    \sum_{\substack{S \subseteq [k]\\ \abs{S}=\ell}} \prod_{i\in S} \alpha_i
    \le (e/\ell)^\ell \Bigl( \sum_{i=1}^k \alpha_i \Bigr)^{\ell}
    = (e/\ell)^\ell (k\alpha)^{\ell} .
  \]
  Because the function $x \mapsto e^{-2x} x^{\ell}$ is maximized when $x=\ell/2$, it follows that
  \[
    \sum_{\ell=1}^d e^{-2k\alpha} \sum_{\substack{S \subseteq [k]\\ \abs{S}=\ell}} \prod_{i\in S} \alpha_i
    \le \sum_{\ell=1}^d e^{-2k\alpha} (e/\ell)^\ell (k\alpha)^\ell
    \le \sum_{\ell=1}^d e^{-\ell} (e/\ell)^\ell (\ell/2)^\ell
    = \sum_{\ell=1}^d 2^{-\ell}
    \le 1 .
  \]
  Therefore,
  \begin{align*}
    g(\alpha_1, \ldots, \alpha_k)
    &= e^{-2 \sum_{i=1}^k \alpha_i} \sum_{\ell=1}^d \sum_{\substack{S \subseteq [k]\\ \abs{S}=\ell}} \sum_{\substack{z \in [m]^S \\ \sum_i z_i=d}} \prod_{i\in S} \Bigl( \alpha_i \bigl( \ln(1/\alpha_i^{1/z_i}) \bigr)^{z_i/2} \Bigr) \\
    &\le 2^d (cm+1)^{d/2} \sum_{\ell=1}^d e^{-2 k \alpha} \sum_{\substack{S \subseteq [k]\\ \abs{S}=\ell}} \prod_{i\in S} \alpha_i \\
    &\le 2^d (cm+1)^{d/2}.
  \end{align*}
  Plugging this bound into Lemma~\ref{lemma:L1prod}, we have
  \[
    W_{1,d}[f]
    \le \bigl( \sqrt{32 e^3 m} \bigr)^d \cdot \bigl( \sqrt{4 (cm+1)} \bigr)^d
    \le \bigl( 72 (\sqrt{c} \cdot m) \bigr)^d . \qedhere
  \]
  \end{proof}

  We now prove Theorem~\ref{thm:L1prod}.
  Recall that we let $\alpha := \sum_{i=1}^k \alpha_i/k \in (0,1]$ for every $(\alpha_1, \ldots, \alpha_k) \in (0,1]^k$.
  We will show that the maximum of the function $g$ defined in Lemma~\ref{lemma:L1prod} is attained at the diagonal $(\alpha, \ldots, \alpha)$.
  We state the claim now and defer the proof to the next section.
  \begin{claim} \label{claim:schurconcave}
    Let $g$ be the function defined in Lemma~\ref{lemma:L1prod}.
    For every $(\alpha_1, \ldots, \alpha_k) \in (0,1]^k$, we have $g(\alpha_1, \ldots, \alpha_k) \le g(\alpha, \ldots, \alpha)$.
  \end{claim}
  
  \begin{proof}[Proof of Theorem~\ref{thm:L1prod}]
    We first apply Claim~\ref{claim:schurconcave} and obtain
  \begin{align*}
    g(\alpha_1, \ldots, \alpha_k)
    \le g(\alpha, \ldots, \alpha)
    = e^{-2 k\alpha} \sum_{\ell=1}^d \sum_{\substack{S \subseteq [k] \\ \abs{S}=\ell}} \alpha^\ell \sum_{\substack{z \in [m]^S \\ \sum_i z_i = d}} \prod_{i \in S} \bigl( \ln \bigl( e/\alpha^{1/z_i} \bigr) \bigr)^{z_i/2} .
  \end{align*}
  We next give an upper bound on $g(\alpha, \ldots, \alpha)$ that has no dependence on the numbers $z_i$.
  By the weighted AM-GM inequality, for every subset $S \subseteq [k]$ of size $\ell$ and numbers $z_i$ such that $\sum_{i \in S} z_i = d$,
  \begin{align*}
    \prod_{i \in S} \bigl( \ln\bigl( e/\alpha^{1/z_i} \bigr) \bigr)^{z_i/2}
    &\le \Bigl( \sum_{i\in S} \frac{z_i \ln\bigl( e/\alpha^{1/z_i} \bigr)}{d} \Bigr)^{d/2} \\
    &= \Bigl( \frac{1}{d} \sum_{i\in S} z_i \Bigl(1 + \frac{1}{z_i} \ln(1/\alpha) \Bigr) \Bigr)^{d/2} \\
    &= \Bigl( 1 + \frac{\ell}{d} \ln(1/\alpha) \Bigr)^{d/2} \\
    &= \bigl( \ln \bigl(e/\alpha^{\ell/d} \bigr) \bigr)^{d/2} .
  \end{align*}
  For every subset $S \subseteq [k]$, the set $\{z \in [m]^S: \sum_i z_i = d\}$ has size at most $\binom{d-1}{\abs{S}-1} \le 2^d$.
  Thus,
  \begin{align*}
    g(\alpha, \ldots, \alpha)
    &\le e^{-2 k\alpha} \sum_{\ell=1}^d \sum_{\substack{S \subseteq [k] \\ \abs{S}=\ell}} \alpha^\ell \sum_{\substack{z \in [m]^S \\ \sum_i z_i = d}} \bigl( \ln\bigl(e/\alpha^{\ell/d} \bigr) \bigr)^{d/2} \\
    &\le 2^d \sum_{\ell=1}^d  e^{-2 k\alpha} \sum_{\substack{S \subseteq [k] \\ \abs{S}=\ell}} \alpha^\ell \bigl( \ln\bigl(e/\alpha^{\ell/d} \bigr) \bigr)^{d/2} \\
    &\le 2^d \sum_{\ell=1}^d  e^{-2 k\alpha} \Bigl( \frac{e k\alpha}{\ell} \Bigr)^\ell \bigl( \ln\bigl(e/\alpha^{\ell/d} \bigr) \bigr)^{d/2} . \numberthis \label{eqn:L1prodc}
  \end{align*}
  For every $\ell \in [k]$, define $g_\ell\colon (0,1] \to \R$ to be
  \[
    g_\ell(x)
    := e^{-2 kx} \Bigl( \frac{e kx}{\ell} \Bigr)^\ell \bigl( \ln\bigl(e/x^{\ell/d} \bigr) \bigr)^{d/2} .
  \]
  We now bound above the maximum of $g_\ell$ over $x \in (0,1]$.
  One can verify easily that the derivative of $g$ is
  \[
    g_\ell'(x)
    = \frac{g_\ell(x)}{2x \ln \bigl(e/x^{\ell/d} \bigr)} \bigl( \ln(1/x^{2\ell/d}) (\ell - 2kx) + (\ell - 4kx) \bigr) .
  \]
  Observe that when $x \le \ell/4k$, then $g_\ell'(x) \ge \frac{g_\ell(x)}{4x \ln (e/x^{\ell/d})} \bigl( \ell \ln(1/x^{2\ell/d}) \bigr) \ge 0$.
  Likewise, when $x \ge \ell/2k$, then $g_\ell'(x) \le \frac{g_\ell(x)}{2x \ln(e/x^{\ell/d})} (-\ell) \le 0$.
  Also, we have $g_\ell(0) = 0$.
Hence, $g_\ell(x) \le g_\ell(\beta_\ell \ell/4k)$ for some $\beta_\ell \in [1,2]$, which is at most
  \[
    e^{-\ell/2} \cdot (e/2)^\ell \cdot \Bigl( \ln\bigl(e (4k/\ell)^{\ell/d} \bigr) \Bigr)^{d/2} .
  \]
  (In the case when $\ell/4k \ge 1$, we have $g_\ell(x) \le g_\ell(1) \le e^{-2k}(ek/\ell)^\ell$.)
  Therefore, plugging this back into Equation~\eqref{eqn:L1prodc},
  \begin{align*}
    g(\alpha, \ldots, \alpha)
    \le 2^d \sum_{\ell=1}^d g_\ell(\alpha) 
    \le 2^d \sum_{\ell=1}^d g_\ell(\beta_\ell \ell/4k) 
    &\le 2^d \sum_{\ell=1}^d e^{-\ell/2} \cdot (e/2)^\ell \cdot \Bigl( \ln\bigl(e (4k/\ell)^{\ell/d} \bigr) \Bigr)^{d/2} \\
    &\le 2^d \bigl(e \ln(4ek) \bigr)^{d/2} \sum_{\ell=1}^d 2^{-\ell}  \\
    &\le \bigl( \sqrt{4e\ln(4ek)} \bigr)^d .
  \end{align*}
  Putting this back into the bound in Lemma~\ref{lemma:L1prod}, we conclude that
  \[
    W_{1,d}[f] \le \bigl( 84 \sqrt{m \ln(4ek)} \bigr)^d ,
  \]
  proving the theorem.
\end{proof}

\subsection{Schur-concavity of $g$}

We prove Claim~\ref{claim:schurconcave} in this section.
First recall that the function $g\colon (0,1]^k \to \R$ is defined as
\[
	g(\alpha_1, \ldots, \alpha_k)
	:= \sum_{\ell=1}^d \sum_{\substack{S \subseteq[k] \\ \abs{S}=\ell}} \sum_{\substack{z \in [m]^S \\ \sum_i z_i = d}} \prod_{i \in S} \phi_{z_i}(\alpha_i) ,
\]
where for every positive integer $z$, the function $\phi_z \colon (0,1] \to \R$ is defined by
\[
  \phi_z(x) = x \ln(e/x^{1/z})^{z/2} .
\]

The proof of Claim~\ref{claim:schurconcave} follows from showing that $g$ is {\em Schur-concave}.
Before defining it, we first recall the concept of majorization.
Let $x, y \in \R^k$ be two vectors.
We say that $y$ {\em majorizes} $x$, denoted by $x \prec y$, if for every $j \in [k]$ we have
\[
  \sum_{i=1}^j x_{(i)} \le \sum_{i=1}^j y_{(i)} ,
\]
and $\sum_{i=1}^k (x_i - y_i) = 0$, where $x_{(i)}$ and $y_{(i)}$ are the $i$th largest coordinates in $x$ and $y$ respectively.

A function $f\colon D \to \R$ where $D \subseteq \R^k$ is {\em Schur-concave} if whenever $x \prec y$ we have $f(x) \ge f(y)$.
We will show that $g$ is Schur-concave using the Schur--Ostrowski criterion.

\begin{theorem}[Schur--Ostrowski criterion (Theorem~12.25 in~\cite{PevcaricPT92})] \label{thm:schur-ostrowski}
Let $f\colon D \to \R$ be a function where $D \subseteq \R^k$ is permutation-invariant, and assume that the first partial derivatives of $f$ exist in $D$.
Then $f$ is Schur-concave in $D$ if and only if
\[
	(x_j - x_i) \Bigl( \frac{\partial f}{\partial x_i} - \frac{\partial f}{\partial x_j} \Bigr) \ge 0
\]
for every $x \in D$, and every $1 \le i \neq j \le k$.
\end{theorem}

Claim~\ref{claim:schurconcave} then follows from the observation that $(\sum_i x_i/k, \ldots, \sum_i x_i/k) \prec x$ for every $x \in [0,1]^k$.

\begin{claim} \label{claim:derivatives}
For every $x \in (0,1]$ we have
\begin{enumerate}
	\item $\phi_z(x) \ge 0$;
	\item $\phi_z'(x) = \frac{1}{2} \ln \bigl(\frac{e}{x^{2/z}}\bigr) \ln\bigl(\frac{e}{x^{1/z}}\bigr)^{z/2-1} > 0$, and
	\item $\phi_z''(x) = -\frac{1}{2xz} \ln\bigl(\frac{e}{x^{1/z}}\bigr)^{z/2-2}\bigl( 2\ln\bigl(\frac{e}{x^{1/z}}\bigr) + (\frac{z}{2}-1) \ln\bigl(\frac{e}{x^{2/z}}\bigr) \bigr) \le 0$.
\end{enumerate}
\end{claim}
\begin{proof}
The derivatives of $\phi_z$ and the non-negativity of $\phi_z$ and $\phi_z'$ can be verified easily.
It is also clear that $\phi_z''$ is non-positive when $z \ge 2$.
Thus it remains to verify $\phi_1''(x) \le 0$ for every $x$.
We have
\[
	\phi_1''(x) = -\frac{1}{2x} \ln\Bigl(\frac{e}{x}\Bigr)^{-3/2} \Bigl( 2\ln\Bigl(\frac{e}{x}\Bigr) - \frac{1}{2} \ln\Bigl(\frac{e}{x^2}\Bigr)  \Bigr) .
\]
It follows from $\frac{1}{2} \ln(e/x^2) \le \ln(e^2/x^2) = 2\ln(e/x)$ that $\phi_1''(x) \le 0$.
\end{proof}

\begin{lemma} \label{lemma:schurconcave}
  $g$ is Schur-concave.
\end{lemma}
\begin{proof}
  Fix $1 \le u \neq v \le k$ and write $g = g_1 + g_2$, where
  \[
  	g_1(\alpha_1, \ldots, \alpha_k)
  	:= \sum_{\ell=1}^d \sum_{\substack{S \subseteq[k], \abs{S}=\ell \\ (S \ni u \wedge S \not\ni v) \vee (S \not\ni u \wedge S \ni v)}} \sum_{\substack{z \in [m]^S \\ \sum_i z_i = d}} \prod_{i \in S} \phi_{z_i}(\alpha_i)
  \]
  and
  \[
  	g_2(\alpha_1, \ldots, \alpha_k)
  	:= \sum_{\ell=1}^d \sum_{\substack{S \subseteq[k], \abs{S}=\ell \\ (S \ni u \wedge S \ni v) \vee (S \not\ni u \wedge S \not\ni v)}} \sum_{\substack{z \in [m]^S \\ \sum_i z_i = d}} \prod_{i \in S} \phi_{z_i}(\alpha_i) .
  \]
  
  We will show that for every $\alpha \in (0,1]^k$, whenever $\alpha_v \le \alpha_u$ we have (1) $\Bigl(\frac{\partial g_1}{\partial \alpha_u} - \frac{\partial g_1}{\partial \alpha_v}\Bigr)(\alpha) \le 0$ and (2) $\Bigl(\frac{\partial g_2}{\partial \alpha_u} - \frac{\partial g_2}{\partial \alpha_v}\Bigr)(\alpha) \le 0$, from which the lemma follows from Theorem~\ref{thm:schur-ostrowski}.

  For $g_1$, since $\phi_z'' \le 0$ and $\alpha_v \le \alpha_u$, we have $\phi_{z_u}'(\alpha_v) \ge \phi_{z_u}'(\alpha_u)$.
  Moreover, as $\phi_z \ge 0$ and $\phi_z' > 0$, we have
  \begin{align*}
  	\frac{\partial g_1}{\partial \alpha_u}(\alpha) 
  	&\le \sum_{\ell=1}^d \sum_{\substack{S \subseteq[k], \abs{S}=\ell \\ (S \ni u \wedge S \not\ni v)}} \sum_{\substack{z \in [m]^S \\ \sum_i z_i = d}} \prod_{\substack{i \in S \\ i \neq u}} \phi_{z_i}(\alpha_i) \cdot \phi_{z_u}'(\alpha_u) \cdot \frac{\phi_{z_u}'(\alpha_v)}{\phi_{z_u}'(\alpha_u)} \\
   &= \sum_{\ell=1}^d \sum_{\substack{S \subseteq[k], \abs{S}=\ell \\ (S \ni u \wedge S \not\ni v)}} \sum_{\substack{z \in [m]^S \\ \sum_i z_i = d}} \prod_{\substack{i \in S \\ i \neq u}} \phi_{z_i}(\alpha_i) \cdot \phi_{z_u}'(\alpha_v) \\
  	&= \sum_{\ell=1}^d \sum_{\substack{S \subseteq[k], \abs{S}=\ell \\ (S \ni v \wedge S \not\ni u)}} \sum_{\substack{z \in [m]^S \\ \sum_i z_i = d}} \prod_{\substack{i \in S \\ i \neq v}} \phi_{z_i}(\alpha_i) \cdot \phi_{z_v}'(\alpha_v) 
  	= \frac{\partial g_1}{\partial \alpha_v}(\alpha) ,
  \end{align*}
  where in the second equality we simply renamed $z_u$ to $z_v$.

  We now show that $\Bigl(\frac{\partial g_2}{\partial \alpha_u} - \frac{\partial g_2}{\partial \alpha_v}\Bigr)(\alpha) \le 0$ whenever $\alpha_v \le \alpha_u$.  
  For all positive integers $z$ and $w$, define $\psi_{z,w}\colon (0,1]^2 \to \R$ by
  \[
  	\psi_{z,w}(x,y) := \phi_z'(x)\phi_w(y) + \phi_w'(x)\phi_z(y) - \phi_z(x)\phi_w'(y) - \phi_w(x)\phi_z'(y) .
  \]
  Note that when $x = y$ we have $\psi_{z,w}(x,x) = 0$.
  Moreover, when $z = w$ we have $\psi_{z,z}(x,y) = 2 (\phi_z'(x)\phi_z(y) - \phi_z(x)\phi_z'(y))$.
  For every $x, y \in (0,1]$, by Claim~\ref{claim:derivatives} we have
  \[
  	\frac{\partial}{\partial y} \psi_{z,w}(x,y)
  	= \phi_z'(x)\phi_w'(y) + \phi_w'(x)\phi_z'(y) - \phi_z(x)\phi_w''(y) - \phi_w(x)\phi_z''(y)
  	\ge 0 .
  \]
  Since $\psi_{z_u,z_v}(\alpha_u,\alpha_u) = 0$, we have $\psi_{z_u,z_v}(\alpha_u,\alpha_v) \le 0$ whenever $\alpha_v \le \alpha_u$, and so
  \begin{multline*}
  	\Bigl( \frac{\partial g_2}{\partial \alpha_u} - \frac{\partial g_2}{\partial \alpha_v} \Bigr)(\alpha) = \\
  	\sum_{\ell=2}^d \sum_{\substack{S \subseteq [k] \\ \abs{S}=\ell \\ S \ni u \wedge S \ni v}} \Bigl( \sum_{\substack{z \in [m]^S \\ \sum_i z_i = d \\ z_u = z_v}} \prod_{\substack{i\in S \\ i \neq u \\ i \neq v}} \phi_{z_i}(\alpha_i) \cdot \psi_{z_u,z_v}(\alpha_u,\alpha_v)/2 + \sum_{\substack{z \in [m]^S \\ \sum_i z_i = d \\ z_u < z_v}} \prod_{\substack{i\in S \\ i \neq u \\ i \neq v}} \phi_{z_i}(\alpha_i) \cdot \psi_{z_u,z_v}(\alpha_u,\alpha_v) \Bigr) \le 0
  \end{multline*}
  because the values $\phi_{z_i}$ are non-negative.
\end{proof}

\subsection{Lower bound}

In this section we prove Claim~\ref{claim:lowerbound}.
We first restate our claim.

\lowerbound*

\begin{proof}
  Let $k = d \cdot 2^m$ and $f_1, \ldots, f_k\colon \B^{mk} \to \B$ be the OR function on $k$ disjoint sets of $m$ bits.
It is easy to verify that $\hf_i(\emptyset) = 1 - 2^{-m}$ and $\abs{\hf_i(S)} = 2^{-m}$ for every $S \neq \emptyset$.
Consider the product test $f := \prod_{i=1}^k f_i$.
Using the fact that $1 - x \ge e^{-x(1+x)}$ for $x \in [0,1/2]$, we have
\[
  (1 - 2^{-m})^k
  \ge e^{-2^m (1 + 2^{-m}) k}
  \ge e^{-d (1 + 2^{-m})}
  \ge e^{-3d/2}.
\]
Hence,
\begin{align*}
  W_{1,d}[f]
  &= \sum_{\substack{z \in \{0,\ldots,m\}^k \\ \sum_i z_i = d}} \prod_{i=1}^k W_{z_i}[f_i] \\
  &\ge \sum_{\abs{S}=d} \Bigl( \prod_{i \in S} W_{1,1}[f_i] \prod_{i \not\in S} W_{1,0}[f_i] \Bigr) \\
  &= \binom{k}{d} \cdot (m 2^{-m})^d \cdot (1 - 2^{-m})^{k-d} \\
  &\ge \Bigl(\frac{d \cdot 2^m}{d}\Bigr)^d \cdot (m2^{-m})^d \cdot e^{-3d/2} \\
  &= (m/e^{3/2})^d . \qedhere
\end{align*}
\end{proof}
\section{Pseudorandom generators} \label{sec:prg}

In this section, we use Theorem~\ref{thm:L1prod} to construct two pseudorandom generators for product tests.
The first one (Theorem~\ref{thm:prg-prod}) has seed length $\tilde O(m + \log(k/\eps)) \log(k/\eps)$.
The second one (Theorem~\ref{thm:prg-xor}) has a seed length of $\tilde O(m + \log(n/\eps))$ but only works for product tests with outputs $\pmone$ and their variants (see Corollary~\ref{cor:prg-L1}).
We note that Theorem~\ref{thm:prg-xor} can also be obtained using Theorem~\ref{thm:L1prod-Bool} in place of Theorem~\ref{thm:L1prod}.

Both constructions use the Ajtai--Wigderson framework~\cite{AjtaiW89,GopalanMRTV12}, and follow from recursively applying the following theorem, which roughly says that $2^{-\tilde\Omega(m + \log(k/\eps))}$-almost $O(m + \log(k/\eps))$-wise independence plus constant fraction of noise fools product tests.

\bipnfp*

Theorem~\ref{thm:bipnfp} follows immediately by combining Theorem~\ref{thm:L1prod} and Lemma~\ref{lemma:ForbesK18} below.

\begin{lemma} \label{lemma:ForbesK18}
  Let $f\colon \B^n \to [-1,1]$ be a product test with $k$ functions of input length $m$.
  Let $d$ be a positive integer.
  Let $D, T, U$ be a $\delta$-almost $(d+m)$-wise independent, a $\gamma$-almost $(d+m)$-wise independent, and the uniform distributions over $\B^n$, respectively.
  Then
  \[
    \bigl| \E[f(D + T \wedge U)] - \E[f(U)] \bigr| \le k \cdot \bigl( \sqrt{\delta} \cdot W_{1,\le d+m}[f] + 2^{-d/2} + \sqrt{\gamma} \bigr) ,
  \]
  where ``$+$'' and ``$\wedge$'' are bit-wise XOR and AND respectively.
\end{lemma}

\begin{proof}
   We slightly modify the decomposition in~\cite[Proposition~6.1]{ForbesK18} as follows.
   Let $f$ be a product test and write $f = \prod_{i=1}^k f_i$.
   As the distribution $D + T \wedge U$ is symmetric, we can assume the function $f_i$ is defined on the $i$th $m$ bits.
   For every $i \in \{1, \ldots, k\}$, let $f^{\le i} = \prod_{j \le i} f_j$ and $f^{>i} = \prod_{j > i} f_j$.
   We decompose $f$ into
    \begin{align}
      f = \hf_{\emptyset} + L + \sum_{i=1}^k H_i f^{> i} , \label{eqn:FKdecomposition}
    \end{align}
    where
    \[
      L := \sum_{\substack{\alpha \in \B^{mk} \\ 0 < \abs{\alpha} < d}} \hf_{\alpha} \chi_{\alpha}
    \]
    and
    \[
      H_i := \sum_{\substack{\alpha = (\alpha_1, \ldots, \alpha_i) \in \B^{mi} : \\ 
      \text{the $d$th $1$ in $\alpha$ appears in $\alpha_i$} }} \hf^{\le i}_{\alpha} \chi_{\alpha} . 
    \]
		We now show that the expressions on both sides of Equation~\eqref{eqn:FKdecomposition} are identical.
		Clearly, every Fourier coefficient on the right hand side is a coefficient of $f$.
    To see that every coefficient of $f$ appears on the right hand side exactly once, let $\alpha = (\alpha_1, \dots, \alpha_k) \in \B^{mk}$ and $\hf_\alpha = \prod_{i=1}^k \hf_i(\alpha_i)$ be a coefficient of $f$.
    If $\abs{\alpha} < d$, then $\hf_\alpha$ appears in $\hf_\emptyset$ or $L$.
    Otherwise, $\abs{\alpha} \ge d$.
    Then the $d$th $1$ in $\alpha$ must appear in one of $\alpha_1, \ldots, \alpha_k$.
    Say it appears in $\alpha_i$.
    Then we claim that $\alpha$ appears in $H_i f^{>i}$.
    This is because the coefficient indexed by $(\alpha_1, \ldots, \alpha_i)$ appears in $H_i$,
    and the coefficient indexed by $(\alpha_{i+1}, \ldots, \alpha_k)$ appears in $f^{>i}$.
    Note that all the coefficients in each function $H_i$ have weights between $d$ and $d+m$,
    and because our distributions $D$ and $T$ are both almost $(d+m)$-wise independent, we get an error of $2^{-d} + \gamma$ in Lemma~7.1 in~\cite{ForbesK18}.
    The rest of the analysis follows from~\cite{ForbesK18} or~\cite{HLV-bipnfp}.
\end{proof}

\subsection{Generator for product tests}

We now prove Theorem~\ref{thm:prg-prod}.

\prgprod*

The high-level idea is very simple.
Let $f$ be a product test.
For every choice of $D$ and $T$ in Theorem~\ref{thm:bipnfp}, the function $f'\colon \B^T \to [-1,1]$ defined by $f'(y) := f(D + T \wedge y)$ is also a product test.
So we can apply Theorem~\ref{thm:bipnfp} again and recurse.
In Lemma~\ref{lemma:recursion} below we show that if we repeat this argument for $t = O(\log(k/\eps))$ times with $t$ independent copies of $D$ and $T$, then for every fixing of $D_1, \ldots, D_t$ and with high probability over the choice of $T_1, \ldots, T_t$, the restricted product test defined on $\B^{\bigwedge_{i=1}^t T_i}$ is a product test with $k$ functions of input length $m/2$.
Now we simply repeat above for $O(\log m) = \tilde O(1)$ steps so that $f$ becomes a constant function and we are done.

\begin{lemma} \label{lemma:recursion}
  If there is an explicit generator $G'\colon \B^{\ell'} \to \B^n$ that fools product tests with $k$ functions of input length $m/2$ with error $\eps'$ and seed length $\ell'$, then there is an explicit generator $G\colon \B^{\ell} \to \B^n$ that fools product tests with $k$ functions of input length $m$ with error $\eps' + \eps$ and seed length 
  \[
    \ell' + O(\log(k/\eps)) \bigl( (m + \log(k/\eps) )(\log m + \log\log(k/\eps)) + \log\log n \bigr)
    = \ell' + \tilde O(m + \log(k/\eps)) \log(k/\eps) .
  \]
\end{lemma}
\begin{proof}
  Let $C$ be a sufficiently large constant.
  Let $d = C(m + \log(k/\eps))$, $\delta = d^{-2d}$, and $t = C\log(k/\eps)$.
  Let $D_1, \ldots, D_t, T_1, \ldots, T_t$ be $2t$ independent $\delta$-almost $d$-wise independent distributions over $\B^n$.
  Define $D^{(1)} := D_1$ and $D^{(i+1)} := D_{i+1} + T_i \wedge D^{(i)}$.

  Let $D := D^{(t)}$, $T := \bigwedge_{i=1}^t T_i$.
  For a subset $S \subseteq [n]$, define the function $\PAD_S(x)\colon \B^{\abs{S}} \to \B^n$ to output $n$ bits of which the positions in $S$ are the first $\abs{S}$ bits of $x0^{\abs{S}}$ and the rest are $0$.
  Our generator $G$ outputs 
	\[
		D + T \wedge \PAD_T(G') .
	\]

  We first look at the seed length of $G$.
  By~\cite[Lemma~4.2]{NaN93}, sampling the distributions $D_i$ and $T_i$ takes a seed of length
  \begin{align*}
    s 
    &:= t \cdot O(d \log d + \log\log n) \\
    &= t \cdot O \bigl( (m + \log(k/\eps)) (\log m + \log\log(k/\eps)) + \log\log n \bigr) \\
    &= t \cdot \tilde O \bigl( m + \log(k/\eps) \bigr).
  \end{align*}
  Hence the total seed length of $G$ is $\ell' + s = \ell' + \tilde O(m + \log(k/\eps)) \log(k/\eps)$.

  \medskip

  We now look at the error of $G$.
  By our choice of $\delta$ and applying Theorem~\ref{thm:bipnfp} recursively for $t$ times, we have
  \begin{align*}
    \bigl| \E[f(D + T \wedge U)] - \E[f(U)] \bigr|
    &\le t \cdot k \cdot \Bigl( \sqrt{\delta} \cdot \bigl( 170 \cdot \sqrt{m \ln(ek)} \bigr)^d + 2^{-(d - m)/2} \Bigr) \\
    &\le t \cdot k \cdot \Bigl( \Bigl( \frac{170 \sqrt{m \ln(ek)}}{d} \Bigr)^d + 2^{-\Omega(d)} \Bigr) \\
    &\le t \cdot 2^{-\Omega(d)}
    \le \eps/2 .
  \end{align*}
  Next, we show that for every fixing of $D$ and most choices of $T$, the function $f_{D,T}(y) := f(D + T \wedge y)$ is a product test with $k$ functions of input length $m/2$, which can be fooled by $G'$.
  
  Because the variables $T_i$ are independent and each of them is $\delta$-almost $d$-wise independent, for every subset $I \subseteq [n]$ of size at most $m \le d$, we have
  \[
    \Pr \bigl[ \abs{T \cap I} \ge m/2 \bigr]
    \le \binom{\abs{I}}{m/2} (2^{-m/2} + \delta)^t
    \le 2^m \cdot 2^{-\Omega(mt)} 
    \le \eps/2k.
  \]
  It follows by a union bound over the $k$ subsets $I_1, \ldots, I_k$ that for every fixing of $D$, with probability at least $1 - \eps/2$ over the choice of $T$, the function $f_{D,T}$ is a product test with $k$ functions of input length $m/2$, which can be fooled by $G'$ with error $\eps'$.
  Hence $G$ fools $f$ with error $\eps' + \eps$.
\end{proof}

\begin{proof}[Proof of Theorem~\ref{thm:prg-prod}]
  We apply Lemma~\ref{lemma:recursion} recursively for $r := O(\log m) = \tilde O(1)$ times.
  Note that a product test of input length $0$ is a constant function, which can always be fooled with zero error.
  Hence we have a generator that fools product tests with $k$ functions of input length $m$, with error $r \cdot \eps$ and seed length
  \[
    r \cdot O \bigl( (m + \log(k/\eps)) (\log m + \log\log(k/\eps)) + \log\log n \bigr) \log(k/\eps)
    = \tilde O(\log(k/\eps) + m) \log(k/\eps) . \qedhere
	\]
	Replacing $\eps$ with $\eps/r$ proves the theorem.
\end{proof}

\subsection{Almost-optimal generator for XOR of Boolean functions}

In this section, we construct our generator for product tests with outputs $\pmone$, which correspond to the XOR of Boolean functions $f_i$ defined on disjoint inputs.
Throughout this section we will call these tests {\em $\pmone$-products}.
We first restate our theorem.

\prgxor*

Theorem~\ref{thm:prg-xor} relies on applying the following lemma recursively in different ways.
From now on, we will relax our tests to allow one of the $k$ functions to have input length greater than $m$, but bounded by $O(m + \log(n/\eps))$.

\begin{lemma} \label{lemma:xor-recursion}
	There exists a constant $C$ such that the following holds.
  Let $m$ and $s$ be two integers such that $m \ge C \log\log(n/\eps)$ and $s = 5(m + \log(n/\eps))$.
  If there is an explicit generator $G'\colon \B^{\ell'} \to \B^n$ that fools $\pmone$-products with $k' \le 16^{m+1}$ functions, $k'-1$ of which have input lengths $\le m/2$ and one has length $\le s$, with error $\eps'$ and seed length $\ell'$, then there is an explicit generator $G\colon \B^{\ell} \to \B^n$ that fools $\pmone$-products with $k \le 16^{2m+1}$ functions, $k-1$ of which have input lengths $\le m$ and one has length $\le s$, with error $\eps' + \eps$ and seed length $\ell = \ell' + O(m + \log(n/\eps)) (\log m + \log\log(n/\eps)) = \ell' + \tilde O(m + \log(n/\eps))$.
\end{lemma}

The proof of Lemma~\ref{lemma:xor-recursion} closely follows a construction by Meka, Reingold and Tal~\cite{MekaRT18}. 
First of all, we will use the following generator in~\cite{MekaRT18}.
It fools any $\pmone$-products when the number of functions $k$ is significantly greater than the input length $m$ of the functions $f_i$.

\begin{lemma}[Lemma~6.2 in~\cite{MekaRT18}] \label{lemma:XOR-many}
  There exists a constant $C$ such that the following holds.
  Let $n, k, m, s$ be integers such that $C \log\log(n/\eps) \le m \le \log n$ and $16^m \le k \le 2 \cdot 16^{2m}$.
  There exists an explicit pseudorandom generator $G_{\oplus \Many}\colon \B^\ell \to \B^n$ that fools $\pmone$-products with $k$ non-constant functions, $k-1$ of which have input lengths $\le m$ and one has length $\le s$, with error $\eps$ and seed length $O(s + \log(n/\eps))$.
\end{lemma}

Here is the high-level idea of proving Lemma~\ref{lemma:xor-recursion}.
We consider two cases depending on whether $k$ is large with respect to $m$.
If $k \ge 16^m$, then by Lemma~\ref{lemma:XOR-many}, the generator $G_{\xor \Many}$ fools $f$.
Otherwise, we show that for every fixing of $D$ and most choices of $T$, the restriction of $f$ under $(D,T)$ is a $\pmone$-product with $k$ functions, $k-1$ of which have input length $\le m/2$ and one has length $\le s$.
More specifically, we will show that for most choices of $T$, the following would happen:
for the function with input length $\le s$, at most $s/2$ of its inputs remain in $T$;
for the rest of the functions with input length $\le m$, after being restricted by $(D,T)$, at most $\ceil{s/2m}$ of them have input length $> m/2$, and so they are defined on a total of $s/2$ positions in $T$.
Now we can think of these ``bad'' functions as one function with input length $\le s$, and the rest of the at most $k$ ``good'' functions have input length $m/2$.
So we can apply the generator $G'$ in our assumption.

\begin{proof} [Proof of Lemma~\ref{lemma:xor-recursion}]
	Let $C$ be the constant in Lemma~\ref{lemma:XOR-many} and $C'$ be a sufficiently large constant.

	Let $d = C's$ and $\delta = d^{-2d}$.
  Let $D_1, \ldots, D_{50}, T_1, \ldots, T_{50}$ be $100$ independent $\delta$-almost $d$-wise independent distributions over $\B^n$.
  Define $D^{(1)} := D_1$ and $D^{(i+1)} := D_{i+1} + T_i \wedge D^{(i)}$.

  Let $D := D^{(50)}$, $T := \bigwedge_{i=1}^{50} T_i$ and $G_{\oplus \Many}$ be the generator in Lemma~\ref{lemma:XOR-many} with respect to the values of $n, k, m, s$ given in this lemma.  
  For a subset $S \subseteq [n]$, define the function $\PAD_S(x)\colon \B^{\abs{S}} \to \B^n$ to output $n$ bits of which the positions in $S$ are the first $\abs{S}$ bits of $x0^{\abs{S}}$ and the rest are $0$.
  Our generator $G$ outputs 
	\[
		(D + T \wedge \PAD_T(G')) + G_{\oplus \Many} .
	\]

  We first look at the seed length of $G$.
  By Lemma~\ref{lemma:XOR-many}, $G_{\oplus \Many}$ uses a seed of length $O(s + \log(n/\eps)) = O(m + \log(n/\eps))$.
  By~\cite[Lemma~4.2]{NaN93}, sampling the distributions $D_i$ and $T_i$ takes a seed of length
  \[
    O(s \log s)
    = O \bigl( m + \log(n/\eps )\bigr) (\log m + \log\log(n/\eps))
    = \tilde O(m + \log(n/\eps)) .
  \]
  Hence the total seed length of $G$ is $\ell' + O( m + \log(n/\eps)) (\log m + \log\log(n/\eps)) = \ell' + \tilde O(m + \log(n/\eps))$.
  
  \medskip

  We now show that $G$ fools $f$.
	Write $f = \prod_{i=1}^k f_i$, where $f_i\colon\B^{I_i} \to \pmone$.
	Without loss of generality we can assume each function $f_i$ is non-constant.
	We consider two cases.

  \paragraph{$k$ is large:} If $k \ge 16^m$, then for every fixing of $D$, $T$ and $G'$, the function $f'(y) := f(D + T \wedge \PAD_T(G') + y)$ is also a $\pmone$-product with the same parameters as $f$.
	Note that we always have $k \le n$ and so $m \le \log n$.
  Hence it follows from Lemma~\ref{lemma:XOR-many} that the generator $G_{\oplus \Many}$ fools $f'$ with error $\eps$.
  Averaging over $D$, $T$ and $G'$ shows that $G$ fools $f$ with error $\eps$.

  \paragraph{$k$ is small:} Now suppose $k \le 16^m$.
  For every fixing of $G_{\oplus \Many}$, consider $f'(y) := f(y + G_{\oplus \Many})$.
  Again, $f'$ is a $\pmone$-product with the same parameters as $f$.
	In particular, it is a $\pmone$-product with $k$ functions with input length $s$.
  So, by our choice of $\delta$ and applying Theorem~\ref{thm:bipnfp} recursively for $50$ times, we have
  \begin{align*}
    \bigl| \E[f'(D + T \wedge U)] - \E[f'(U)] \bigr|
    &\le 50 \cdot k \cdot \Bigl( \sqrt{\delta} \cdot \bigl( 170 \cdot \sqrt{s \ln(ek)} \bigr)^d + 2^{-(d-s)/2} \Bigr) \\
    &\le 50 \cdot 2^s \cdot \Bigl( (170s/d)^d + 2^{-\Omega(s)} \Bigr) \\
    &\le 2^{-\Omega(s)}
    \le \eps/2 .
  \end{align*}
  Next, we show that for every fixing of $D$ and most choices of $T$, the function $f'_{D,T}(y) := f'(D + T \wedge y)$ is a $\pmone$-product with $k$ functions, $k-1$ of which have input lengths $\le m/2$ and one has length $\le s$, which can be fooled by $G'$.
  
  Because the variables $T_i$ are independent and each of them is $\delta$-almost $d$-wise independent, for every subset $I \subseteq [n]$ of size at most $d$, we have
  \[
    \Pr[T \cap I = I]
    = \prod_{i=1}^{50} \Pr[T_i \cap I = I]
    \le (2^{-\abs{I}} + \delta)^{50}
    \le (3/4)^{-50\abs{I}} .
  \]

	Without loss of generality, we assume $I_1, \ldots, I_{k-1}$ are the subsets of size at most $m$ and $I_k$ is the subset of size at most $s$.
	We now look at which subsets $T \cap I_i$ have length at most $m/2$ and which subsets do not.
	For the latter, we collect the indices in these subsets.

  Let $G:= \{i \in [k-1]: \abs{T \cap I_i} \le m/2\}$, $B:= \{i \in [k-1]: \abs{T \cap I_i} > m/2\}$ and $BV := \{j \in [n]: j \in \bigcup_{i\in B} (T \cap I_i) \}$.
  We claim that with probability $1 - \eps/2$ over the choice of $T$, we have $\abs{BV} \le s$.
  Note that the indices in $BV$ either come from $I_k$, or $I_i$ for $i \in [k-1]$.
  For the first case, the probability that at least $s/2$ of the indices in $I_k$ appear in $BV$ is at most
  \[
    \binom{\abs{I_k}}{s/2} (3/4)^{-25s}
    \le 2^s \cdot (3/4)^{-25s}
    \le \eps/4 .
  \]
  For the second case, note that if at least $s/2$ of the variables in $\bigcup_{i \in [k-1]} I_i$ appear in $BV$, then
  they must appear in at least $\ceil{s/2m}$ of the subsets $T \cap I_1, \ldots, T \cap I_{k-1}$.
  The probability of the former is at most the probability of the latter, which is at most 
  \[
    \binom{k-1}{\ceil{s/2m}} \binom{m \cdot \ceil{s/2m}}{s/2} (3/4)^{-25s}
    \le 16^{m \cdot (s/2m+1)} \cdot 2^{m \cdot (s/2m+1)} \cdot (3/4)^{-25s}
    \le \eps/4,
  \]
  because $k \le 16^m$ and $m \le s$.
  Hence with probability $1-\eps/2$ over the choice of $T$, the function $f'_{D,T}$ is a product $g \cdot h$, where $g$ is a product of $\abs{G} \le k-1$ functions of input length $m/2$, and $h$ is a product of $\abs{B}+1$ functions defined on a total of $\abs{BV} \le s$ bits.
  Recall that $k \le 16^m$, so by our assumption $G'$ fools $f'_{D,T}$ with error $\eps'$.
  Therefore $G$ fools $f$ with error $\eps + \eps'$.
\end{proof}

We obtain Theorem~\ref{thm:prg-xor} by applying Lemma~\ref{lemma:xor-recursion} repeatedly in different ways.

\begin{proof}[Proof of Theorem~\ref{thm:prg-xor}]
	Given a $\pmone$-product $f\colon \B^n \to \pmone$ with $k$ functions of input length $m$,
  we will apply Lemma~\ref{lemma:xor-recursion} in stages.
  In each stage, we start with a $\pmone$-product $f$ with $k_1$ functions, $k_1-1$ of which have input lengths $\le m_1 = \max\{m, 2\log(n/\eps)\}$ and one has length $\le s:= 5(m + \log(n/\eps))$.
	Note that $k_1 \le 16^{2 m_1 + 1}$.
	Let $C$ be the constant in Lemma~\ref{lemma:xor-recursion}.
  We apply Lemma~\ref{lemma:xor-recursion} for $t = O(\log m_1)$ times until $f$ is restricted to a $\pmone$-product $f'$ with $k_2$ functions, $k_2 - 1$ of which have input lengths $\le m_2$ and one has length $\le s$, where $m_2 = C \log\log(n/\eps)$, $k_2 \le 16^{2m_2 +1} \le (\log (n/\eps))^r$, and $r := 8C+4$ is a constant.
	This uses a seed of length
  \begin{align*}
    t \cdot O(m + \log(n/\eps)) (\log m + \log\log(n/\eps))
    &\le O(m + \log(n/\eps)) (\log m + \log\log(n/\eps))^2 \\
    &= \tilde O(m + \log(n/\eps)) .
  \end{align*}

  At the end of each stage, we repeat the above argument by grouping every $\ceil{\log (n/\eps) / m_2}$ functions of $f'$ that have input lengths $\le m_2$ as one function of input length $\le 2\log(n/\eps)$, so we can think of $f'$ as a $\pmone$-product with $k_3 := k_2 / \ceil{m_2/(\log n)} \le (\log(n/\eps))^{r-1} \log\log n$ functions, $k_3 - 1$ of which have input lengths $\le \log(n/\eps)$ and one has length $\le s$.

  Repeating above for $r+1 = O(1)$ stages, we are left with a $\pmone$-product of two functions, one has input length $\le C \log \log(n/\eps)$, and one has length $\le s$, which can then be fooled by a $2^{-\Omega(s)}$-biased distribution that can be sampled using $O(m + \log(n/\eps))$ bits~\cite{NaN93}.
  So the total seed length is $O(m + \log(n/\eps)) (\log m + \log\log(n/\eps))^2 = \tilde O(m + \log(n/\eps))$, and the error is $(r+1) \cdot t \cdot \eps$.
  Replacing $\eps$ with $\eps / (r+1)t$ proves the theorem.
\end{proof}

\section{Level-$k$ inequalities} \label{sec:levelk}

In this section, we prove Lemma~\ref{lemma:level-k-inequalities} that gives an upper bound on the $d$th level Fourier weight of a $[0,1]$-valued function in $L_2$-norm.
We first restate the lemma.

\levelk*

Our proof closely follows the argument in~\cite{Talagrand96}.

\begin{claim} \label{claim:Talagrand}
  Let $f\colon \B^n \to \R$ have Fourier degree at most $d$ and $\norm{f}_2 = 1$.
  Let $g\colon \B^n \to [0,1]$ be any function.
  If $t_0 \ge 2e^{d/2}$, then
  \[
    \E \bigl[g(x) \abs{f(x)} \bigr]
		\le \E[g] t_0 + 2e t_0^{1-2/d} e^{-\frac{d}{2e} t_0^{2/d}} .
  \]
\end{claim}

To prove this claim, we will use the following concentration inequality for functions with Fourier degree $k$ from~\cite{DinurFKO07}.

\begin{theorem}[Lemma~2.2 in~\cite{DinurFKO07}]
  Let $f\colon \B^n \to \R$ have Fourier degree at most $d$ and assume that $\norm{f}_2 := \sum_S \hf_S^2 = 1$.  Then for any $t \ge (2e)^{d/2}$,
  \[
    \Pr \bigl[ \abs{f} \ge t \bigr] \le e^{-\frac{d}{2e} t^{2/d}} .
  \]
\end{theorem}

We also need to bound above the integral of $e^{-\frac{d}{2e} t^{2/d}}$.

\begin{claim} \label{claim:incomplete-gamma}
  Let $d$ be any positive integer.
  If $t_0 \ge (2e)^{d/2}$, then we have
  \[
    \int_{t_0}^\infty e^{-\frac{d}{2e} t^{2/d}} dt
    \le 2e t_0^{1-2/d} e^{-\frac{d}{2e} t_0^{2/d}} .
  \]
\end{claim}

\begin{proof}
  First we apply the following change of variable to the integral.
  We set $s = \frac{d}{2e} t^{2/d}$ and obtain
  \[
    \int_{t_0}^\infty e^{-\frac{d}{2e} t^{2/d}} dt
    = e \Bigl( \frac{2e}{d} \Bigr)^{d/2-1} \int_{s_0}^\infty s^{d/2-1} e^{-s} ds ,
  \]
  where $s_0 = \frac{d}{2e} t_0^{2/d}$.
  Define 
	\[
		\Gamma_{s_0}(d) = \int_{s_0}^\infty s^{d-1} e^{-s} ds .
	\]
  (Note that when $s_0 = 0$ then $\Gamma_0(d)$ is the Gamma function.)
  Using integration by parts, we have
  \begin{align}
    \Gamma_{s_0}(d) = s_0^{d-1} e^{-s_0} + (d-1) \Gamma_{s_0}(d-1) . \label{eqn:gamma}
  \end{align}
  Moreover, when $d \le 1$, we have $\Gamma_{s_0}(d) \le s_0^{d-1} \int_{s_0}^\infty e^{-s} ds = s_0^{d-1} e^{-s_0}$.

  Note that if $t_0 \ge (2e)^{d/2}$, then $s_0 \ge d-2$.
  Hence, if we open the recursive definition of $\Gamma_{s_0}(d/2)$ in Equation~\eqref{eqn:gamma}, we have
  \begin{align*}
    \Gamma_{s_0}(d/2)
    &\le e^{-s_0} \sum_{i=0}^{\ceil{\frac{d}{2}}-1} s_0^{d/2-1-i} \prod_{j=1}^i (d/2 - j) \\
    &\le e^{-s_0} s_0^{d/2-1} \sum_{i=0}^{\ceil{\frac{d}{2}}-1} \Bigl(\frac{d/2-1}{s_0}\Bigr)^i \\
    &\le 2 e^{-s_0} s_0^{d/2-1} ,
  \end{align*}
  because the summation is a geometric sum with ratio at most $1/2$.
  Substituting $s_0$ with $t_0$, we obtain
  \begin{align*}
    e \Bigl( \frac{2e}{d} \Bigr)^{d/2-1} \int_{s_0}^\infty s^{d/2-1} e^{-s} ds
    &\le 2 e \Bigl( \frac{2e}{d} \Bigr)^{d/2-1} e^{-s_0} s_0^{d/2-1} \\
    &= 2 e t_0^{1-2/d} e^{-\frac{d}{2e} t_0^{2/d}} . \qedhere
  \end{align*}
\end{proof}

\begin{proof}[Proof of Claim~\ref{claim:Talagrand}]
  We rewrite $\abs{f(x)}$ as $\int_0^{\abs{f(x)}} \1 dt = \int_0^\infty \1(\abs{f(x)} \ge t) dt$ and obtain
  \begin{align*}
    \E_{x \sim \B^n}[ g(x) |f(x)| ]
    &= \E_{x \sim \B^n}\Bigl[ \int_0^{\infty} g(x) \1(\abs{f(x)} \ge t) dt \Bigr] \\
    &\le \E_{x \sim \B^n}\Bigl[ \int_0^{\infty} \min \bigl\{ g(x), \1(\abs{f(x)} \ge t) \bigr\} dt \Bigr] \\
    &= \int_0^\infty \min \Bigl\{ \E[g], \Pr_x[\abs{f(x)} \ge t] \Bigr\} dt \\
    &\le \int_0^{t_0} \E[g] dt + \int_{t_0}^\infty \Pr \bigl[ \abs{f(x)} \ge t \bigr] dt \\
    &\le \E[g] t_0 + \int_{t_0}^\infty e^{-\frac{d}{2e}t^{2/d}} dt.
  \end{align*}
  Since $t_0 \ge (2e)^{d/2}$, by Claim~\ref{claim:incomplete-gamma} this is at most $\E[g] t_0 + 2e t_0^{1-2/d} e^{-\frac{d}{2e} t_0^{2/d}}$.
\end{proof}

\begin{proof}[Proof of Lemma~\ref{lemma:level-k-inequalities}]
  Define $f$ to be $f(x) := \sum_{\abs{S}=d} \hf_S \chi_S(x)$, where $\hf_S = \hg_S \bigl( \sum_{\abs{T}=d} \hg_T^2 \bigr)^{-1/2}$.
  Note that $\norm{f}_2 = 1$, and we have
	\[
		\E[g(x) f(x)]
		= \frac{\sum_S \hg_S \E[g(x) \chi_S(x)]}{ \bigl( \sum_{\abs{T}=d} \hg_T^2 \bigr)^{1/2}}
		= \Bigl( \sum_{\abs{S}=d} \hg_S^2 \Bigr)^{1/2} .
	\]
  Let $t_0 = (2e \ln(e/\E[g]^{1/d}))^{d/2} \ge (2e)^{d/2}$.
  By Claim~\ref{claim:Talagrand},
  \[
    \Bigl( \sum_{\abs{S}=d} \hg_S^2 \Bigr)^{1/2}
    = \E[g(x) f(x)]
    \le \E[g(x) \abs{f(x)}]
    \le \E[g] t_0 + 2e t_0^{1-2/d} e^{-\frac{d}{2e} t_0^{2/d}}.
  \]
  By our choice of $t_0$, the second term is at most
	\[
		2e t_0^{1-2/d} e^{-\frac{d}{2e} t_0^{2/d}}
		\le \left( 2e \ln\Bigl( \frac{e}{\E[g]^{1/d}} \Bigr) \right)^{d/2} \frac{\E[g]}{e^d}
		\le (2/e)^{d/2} \E[g] \ln\Bigl(\frac{e}{\E[g]^{1/d}} \Bigr)^{d/2} ,
	\]
	which is no greater than the first term.
  So
  \[
    \Bigl( \sum_{\abs{S}=d} \hg_S^2 \Bigr)^{1/2}
    \le 2 \E[g] \bigl( 2e \ln(e/\E[g]^{1/d}) \bigr)^{d/2} .
  \]
  and the lemma follows.
\end{proof}

\subsection*{Acknowledgement}
I thank Salil Vadhan for asking about the Fourier spectrum of product tests.
I also thank Andrej Bogdanov, Gil Cohen, Amnon Ta-Shma, Avishay Tal and Emanuele Viola for very helpful conversations.
I am grateful to Emanuele Viola for his invaluable comments on the write-up.

\bibliographystyle{alpha}
\bibliography{OmniBib}

\newcommand{\etalchar}[1]{$^{#1}$}
\def\cprime{$'$}
\begin{thebibliography}{GMR{\etalchar{+}}12}

\bibitem[Aar10]{Aaronson10}
Scott Aaronson.
\newblock {BQP} and the polynomial hierarchy.
\newblock In {\em 42nd ACM Symp.~on the Theory of Computing (STOC)}, pages
  141--150. ACM, 2010.

\bibitem[ABO84]{AjB84}
Mikl{\'o}s Ajtai and Michael Ben-Or.
\newblock A theorem on probabilistic constant depth computations.
\newblock In {\em 16th ACM Symp.~on the Theory of Computing (STOC)}, pages
  471--474, 1984.

\bibitem[Ajt83]{Ajt83}
Mikl{\'o}s Ajtai.
\newblock {$\Sigma \sp{1}\sb{1}$}-formulae on finite structures.
\newblock {\em Annals of Pure and Applied Logic}, 24(1):1--48, 1983.

\bibitem[AKS87]{AKS87}
Mikl{\'o}s Ajtai, J\'{a}nos Koml\'{o}s, and Endre Szemer{\'e}di.
\newblock Deterministic simulation in logspace.
\newblock In {\em 19th ACM Symp.~on the Theory of Computing (STOC)}, pages
  132--140, 1987.

\bibitem[Ama09]{Amano09}
Kazuyuki Amano.
\newblock Bounds on the size of small depth circuits for approximating
  majority.
\newblock In {\em 36th Coll.~on Automata, Languages and Programming (ICALP)},
  pages 59--70. Springer, 2009.

\bibitem[ASWZ96]{ArmoniSWZ96}
Roy Armoni, Michael~E. Saks, Avi Wigderson, and Shiyu Zhou.
\newblock Discrepancy sets and pseudorandom generators for combinatorial
  rectangles.
\newblock In {\em 37th IEEE Symp.~on Foundations of Computer Science (FOCS)},
  pages 412--421, 1996.

\bibitem[AW89]{AjtaiW89}
Miklos Ajtai and Avi Wigderson.
\newblock Deterministic simulation of probabilistic constant-depth circuits.
\newblock {\em Advances in Computing Research - Randomness and Computation},
  5:199--223, 1989.

\bibitem[BPW11]{BogdanovPW11}
Andrej Bogdanov, Periklis~A. Papakonstantinou, and Andrew Wan.
\newblock Pseudorandomness for read-once formulas.
\newblock In {\em IEEE Symp.~on Foundations of Computer Science (FOCS)}, pages
  240--246, 2011.

\bibitem[BV10]{BrodyV10}
Joshua Brody and Elad Verbin.
\newblock The coin problem, and pseudorandomness for branching programs.
\newblock In {\em 51th IEEE Symp.~on Foundations of Computer Science (FOCS)},
  2010.

\bibitem[CGR14]{CohenGR14}
Gil Cohen, Anat Ganor, and Ran Raz.
\newblock Two sides of the coin problem.
\newblock In {\em Workshop on Randomization and Computation (RANDOM)}, pages
  618--629, 2014.

\bibitem[Cha02]{Chang02}
Mei-Chu Chang.
\newblock A polynomial bound in {F}reiman's theorem.
\newblock {\em Duke Math. J.}, 113(3):399--419, 2002.

\bibitem[CHHL18]{ChattopadhyayHHL18}
Eshan Chattopadhyay, Pooya Hatami, Kaave Hosseini, and Shachar Lovett.
\newblock Pseudorandom generators from polarizing random walks.
\newblock {\em Electronic Colloquium on Computational Complexity}, Technical
  Report TR18-015, 2018.
\newblock www.eccc.uni-trier.de/.

\bibitem[CHLT19]{ChattopadhyayHLT19}
Eshan Chattopadhyay, Pooya Hatami, Shachar Lovett, and Avishay Tal.
\newblock Pseudorandom generators from the second fourier level and
  applications to ac0 with parity gates.
\newblock In {\em I{TCS}'19---{P}roceedings of the 2019 {ACM} {C}onference on
  {I}nnovations in {T}heoretical {C}omputer {S}cience}. 2019.

\bibitem[CHRT18]{ChattopadhyayHRT18}
Eshan Chattopadhyay, Pooya Hatami, Omer Reingold, and Avishay Tal.
\newblock Improved pseudorandomness for unordered branching programs through
  local monotonicity.
\newblock In {\em ACM Symp.~on the Theory of Computing (STOC)}, 2018.

\bibitem[CRS00]{ChariRS00}
Suresh Chari, Pankaj Rohatgi, and Aravind Srinivasan.
\newblock Improved algorithms via approximations of probability distributions.
\newblock {\em J. Comput. System Sci.}, 61(1):81--107, 2000.

\bibitem[CSV15]{ChenSV15}
Sitan Chen, Thomas Steinke, and Salil~P. Vadhan.
\newblock Pseudorandomness for read-once, constant-depth circuits.
\newblock {\em CoRR}, abs/1504.04675, 2015.

\bibitem[DETT10]{DeETT10}
Anindya De, Omid Etesami, Luca Trevisan, and Madhur Tulsiani.
\newblock Improved pseudorandom generators for depth 2 circuits.
\newblock In {\em Workshop on Randomization and Computation (RANDOM)}, pages
  504--517, 2010.

\bibitem[DFKO07]{DinurFKO07}
Irit Dinur, Ehud Friedgut, Guy Kindler, and Ryan O'Donnell.
\newblock On the {F}ourier tails of bounded functions over the discrete cube.
\newblock {\em Israel J. Math.}, 160:389--412, 2007.

\bibitem[DHH18]{DoronHH18}
Dean Doron, Pooya Hatami, and William Hoza.
\newblock Near-optimal pseudorandom generators for constant-depth read-once
  formulas.
\newblock 2018.
\newblock ECCC TR18-183.

\bibitem[EGL{\etalchar{+}}98]{EvenGLNV98}
Guy Even, Oded Goldreich, Michael Luby, Noam Nisan, and Boban Velickovic.
\newblock Efficient approximation of product distributions.
\newblock {\em Random Struct. Algorithms}, 13(1):1--16, 1998.

\bibitem[FK18]{ForbesK18}
Michael~A. Forbes and Zander Kelley.
\newblock Pseudorandom generators for read-once branching programs, in any
  order.
\newblock In {\em 47th IEEE Symposium on Foundations of Computer Science
  (FOCS)}, 2018.

\bibitem[GKM15]{GopalanKM15}
Parikshit Gopalan, Daniel Kane, and Raghu Meka.
\newblock Pseudorandomness via the discrete fourier transform.
\newblock In {\em IEEE Symp.~on Foundations of Computer Science (FOCS)}, pages
  903--922, 2015.

\bibitem[GMR{\etalchar{+}}12]{GopalanMRTV12}
Parikshit Gopalan, Raghu Meka, Omer Reingold, Luca Trevisan, and Salil Vadhan.
\newblock Better pseudorandom generators from milder pseudorandom restrictions.
\newblock In {\em IEEE Symp.~on Foundations of Computer Science (FOCS)}, 2012.

\bibitem[GSW16]{GopalanRW16}
Parikshit Gopalan, Rocco~A. Servedio, and Avi Wigderson.
\newblock Degree and sensitivity: tails of two distributions.
\newblock In {\em 31st {C}onference on {C}omputational {C}omplexity}, volume~50
  of {\em LIPIcs. Leibniz Int. Proc. Inform.}, pages Art. No. 13, 23. Schloss
  Dagstuhl. Leibniz-Zent. Inform., Wadern, 2016.

\bibitem[GY14]{GopalanY14}
Parikshit Gopalan and Amir Yehudayoff.
\newblock Inequalities and tail bounds for elementary symmetric polynomials.
\newblock {\em Electronic Colloquium on Computational Complexity {(ECCC)}},
  21:19, 2014.

\bibitem[HLV18]{HLV-bipnfp}
Elad Haramaty, Chin~Ho Lee, and Emanuele Viola.
\newblock Bounded independence plus noise fools products.
\newblock {\em SIAM J. Comput.}, 47(2):493--523, 2018.

\bibitem[HT18]{HatamiT18}
Pooya Hatami and Avishay Tal.
\newblock Pseudorandom generators for low-sensitivity functions.
\newblock In {\em 9th {I}nnovations in {T}heoretical {C}omputer {S}cience},
  volume~94 of {\em LIPIcs. Leibniz Int. Proc. Inform.}, pages Art. No. 29, 13.
  Schloss Dagstuhl. Leibniz-Zent. Inform., Wadern, 2018.

\bibitem[IMR14]{ImpagliazzoMR14}
Russell Impagliazzo, Cristopher Moore, and Alexander Russell.
\newblock An entropic proof of {C}hang's inequality.
\newblock {\em SIAM J. Discrete Math.}, 28(1):173--176, 2014.

\bibitem[IMZ12]{ImpagliazzoMZ12}
Russell Impagliazzo, Raghu Meka, and David Zuckerman.
\newblock Pseudorandomness from shrinkage.
\newblock In {\em IEEE Symp.~on Foundations of Computer Science (FOCS)}, pages
  111--119, 2012.

\bibitem[INW94]{INW94}
Russell Impagliazzo, Noam Nisan, and Avi Wigderson.
\newblock Pseudorandomness for network algorithms.
\newblock In {\em 26th ACM Symp.~on the Theory of Computing (STOC)}, pages
  356--364, 1994.

\bibitem[KK13]{KellerK13}
Nathan Keller and Guy Kindler.
\newblock Quantitative relation between noise sensitivity and influences.
\newblock {\em Combinatorica}, 33(1):45--71, 2013.

\bibitem[KS18]{KoppartyS18}
Swastik Kopparty and Srikanth Srinivasan.
\newblock Certifying polynomials for {$\rm AC^0[\oplus]$} circuits, with
  applications to lower bounds and circuit compression.
\newblock {\em Theory Comput.}, 14:Article 12, 24, 2018.

\bibitem[LSS{\etalchar{+}}18]{LimayeSSTV18}
Nutan Limaye, Karteek Sreenivasiah, Srikanth Srinivasan, Utkarsh Tripathi, and
  S~Venkitesh.
\newblock The coin problem in constant depth: Sample complexity and parity
  gates.
\newblock In {\em Electronic Colloquium on Computational Complexity (ECCC)},
  number TR18--157, 2018.

\bibitem[Lu02]{Lu02}
Chi-Jen Lu.
\newblock Improved pseudorandom generators for combinatorial rectangles.
\newblock {\em Combinatorica}, 22(3):417--433, 2002.

\bibitem[LV17]{LV-rop}
Chin~Ho Lee and Emanuele Viola.
\newblock More on bounded independence plus noise: Pseudorandom generators for
  read-once polynomials.
\newblock In {\em Electronic Colloquium on Computational Complexity (ECCC)},
  volume~24, page 167, 2017.

\bibitem[LV18]{LeeV-coin}
Chin~Ho Lee and Emanuele Viola.
\newblock The coin problem for product tests.
\newblock {\em ACM Trans. Comput. Theory}, 10(3):Art. 14, 10, 2018.

\bibitem[Man95]{Mansour95}
Yishay Mansour.
\newblock An {$O(n^{\log \log n})$} learning algorithm for {DNF} under the
  uniform distribution.
\newblock {\em J. Comput. System Sci.}, 50(3, part 3):543--550, 1995.
\newblock Fifth Annual Workshop on Computational Learning Theory (COLT)
  (Pittsburgh, PA, 1992).

\bibitem[MRT18]{MekaRT18}
Raghu Meka, Omer Reingold, and Avishay Tal.
\newblock Pseudorandom generators for width-3 branching programs.
\newblock {\em arXiv preprint arXiv:1806.04256}, 2018.

\bibitem[Nis92]{Nis92}
Noam Nisan.
\newblock Pseudorandom generators for space-bounded computation.
\newblock {\em Combinatorica}, 12(4):449--461, 1992.

\bibitem[NN93]{NaN93}
Joseph Naor and Moni Naor.
\newblock Small-bias probability spaces: efficient constructions and
  applications.
\newblock {\em SIAM J.~on Computing}, 22(4):838--856, 1993.

\bibitem[NZ96]{NiZ96}
Noam Nisan and David Zuckerman.
\newblock Randomness is linear in space.
\newblock {\em J.~of Computer and System Sciences}, 52(1):43--52, February
  1996.

\bibitem[O'D14]{ODonnell14}
Ryan O'Donnell.
\newblock {\em Analysis of Boolean Functions}.
\newblock Cambridge University Press, 2014.

\bibitem[PPT92]{PevcaricPT92}
Josip~E. Pe\v{c}ari\'{c}, Frank Proschan, and Y.~L. Tong.
\newblock {\em Convex functions, partial orderings, and statistical
  applications}, volume 187 of {\em Mathematics in Science and Engineering}.
\newblock Academic Press, Inc., Boston, MA, 1992.

\bibitem[RS17]{RossmanS17}
Benjamin Rossman and Srikanth Srinivasan.
\newblock Separation of {$\rm AC^0[\oplus]$} formulas and circuits.
\newblock In {\em 44th {I}nternational {C}olloquium on {A}utomata, {L}anguages,
  and {P}rogramming}, volume~80 of {\em LIPIcs. Leibniz Int. Proc. Inform.},
  pages Art. No. 50, 13. Schloss Dagstuhl. Leibniz-Zent. Inform., Wadern, 2017.

\bibitem[RSV13]{ReingoldSV13}
Omer Reingold, Thomas Steinke, and Salil~P. Vadhan.
\newblock Pseudorandomness for regular branching programs via {F}ourier
  analysis.
\newblock In {\em Workshop on Randomization and Computation (RANDOM)}, pages
  655--670, 2013.

\bibitem[ST18]{ServedioT18}
Rocco~A. Servedio and Li-Yang Tan.
\newblock Improved pseudorandom generators from pseudorandom multi-switching
  lemmas.
\newblock {\em CoRR}, abs/1801.03590, 2018.

\bibitem[Ste04]{Steele04}
J.~Michael Steele.
\newblock {\em The {C}auchy-{S}chwarz master class}.
\newblock MAA Problem Books Series. Mathematical Association of America,
  Washington, DC; Cambridge University Press, Cambridge, 2004.

\bibitem[Ste13]{Steinberger13}
John~P. Steinberger.
\newblock The distinguishability of product distributions by read-once
  branching programs.
\newblock In {\em IEEE Conf.~on Computational Complexity (CCC)}, pages
  248--254, 2013.

\bibitem[SV10]{ShV-dec}
Ronen Shaltiel and Emanuele Viola.
\newblock Hardness amplification proofs require majority.
\newblock {\em SIAM J.~on Computing}, 39(7):3122--3154, 2010.

\bibitem[SVW14]{SteinkeVW14}
Thomas Steinke, Salil~P. Vadhan, and Andrew Wan.
\newblock Pseudorandomness and fourier growth bounds for width-3 branching
  programs.
\newblock In {\em Workshop on Randomization and Computation (RANDOM)}, pages
  885--899, 2014.

\bibitem[Tal96]{Talagrand96}
Michel Talagrand.
\newblock How much are increasing sets positively correlated?
\newblock {\em Combinatorica}, 16(2):243--258, 1996.

\bibitem[Tal17]{Tal17}
Avishay Tal.
\newblock Tight bounds on the fourier spectrum of {AC0}.
\newblock In {\em Conf.~on Computational Complexity (CCC)}, pages 15:1--15:31,
  2017.

\bibitem[TX13]{TrevisanX13}
Luca Trevisan and TongKe Xue.
\newblock A derandomized switching lemma and an improved derandomization of
  ac0.
\newblock In {\em Computational Complexity (CCC), 2013 IEEE Conference on},
  pages 242--247. IEEE, 2013.

\bibitem[Val84]{Valiant84-Majority}
Leslie~G. Valiant.
\newblock Short monotone formulae for the majority function.
\newblock {\em J. Algorithms}, 5(3):363--366, 1984.

\bibitem[Vio09]{ViolaBPvsE}
Emanuele Viola.
\newblock On approximate majority and probabilistic time.
\newblock {\em Computational Complexity}, 18(3):337--375, 2009.

\bibitem[Vio14]{Viola-rbd}
Emanuele Viola.
\newblock Randomness buys depth for approximate counting.
\newblock {\em Computational Complexity}, 23(3):479--508, 2014.

\bibitem[Wat13]{Watson13}
Thomas Watson.
\newblock Pseudorandom generators for combinatorial checkerboards.
\newblock {\em Computational Complexity}, 22(4):727--769, 2013.

\end{thebibliography}

\end{document}